\documentclass[a4paper,11pt]{article}
\usepackage[UKenglish]{babel}
\usepackage{amsmath,amsfonts,amssymb,amsthm}
\usepackage{hyperref}
\usepackage{amsbsy}
\usepackage{times}
\usepackage{bm,bbm}
\usepackage{a4wide}
\usepackage[sans]{dsfont}
\usepackage[mathscr]{euscript}
\usepackage{cite}

\usepackage{enumerate}

\theoremstyle{plain}
\newtheorem{theorem}{Theorem}
\newtheorem{proposition}[theorem]{Proposition}

\theoremstyle{remark}
\newtheorem{remark}{Remark}

\theoremstyle{definition}
\newtheorem{definition}{Definition}

\newcommand{\openone}{\mathds{1}}

\newcommand{\Rbb}{\mathbb{R}}
\newcommand{\Cbb}{\mathbb{C}}

\newcommand{\abs}[1]{\left\vert#1\right\vert}

\newcommand{\Acal}{\mathcal{A}}

\newcommand{\Scal}{\mathcal{S}}

\newcommand{\Fscr}{\mathscr{F}}

\newcommand{\Hscr}{\mathscr{H}}

 \newcommand{\Mscr}{\mathscr{M}}

\newcommand{\Sscr}{\mathscr{S}}

\newcommand{\Xscr}{\mathscr{X}}
\newcommand{\Yscr}{\mathscr{Y}}

\newcommand{\Ao}{\mathsf{A}}
\newcommand{\Do}{\mathsf{D}}

\newcommand{\Mo}{\mathsf{M}}
\newcommand{\No}{\mathsf{N}}

\newcommand{\Xo}{\mathsf{X}}
\newcommand{\Yo}{\mathsf{Y}}
\newcommand{\Zo}{\mathsf{Z}}
\newcommand{\Fo}{\mathsf{F}}

\newcommand{\blambda}{\boldsymbol{\lambda}}
\newcommand{\bxi}{\boldsymbol{\xi}}
\newcommand{\btheta}{\boldsymbol{\theta}}
\newcommand{\bi}{\boldsymbol}
\newcommand{\rmi}{\mathrm{i}}
\newcommand{\rme}{\mathrm{e}}
\newcommand{\rmd}{\mathrm{d}}
\newcommand{\Tr}[1]{{\rm Tr}\,#1}

\begin{document}

\title{Entropic measurement uncertainty relations for all the infinite components of a spin vector}

\author{Alberto Barchielli$^{1,2,3}$
and Matteo Gregoratti$^{1}$
\\  \\
$^1$ Politecnico di Milano, Dipartimento di Matematica, \\ {}\quad
piazza Leonardo da Vinci 32, 20133 Milano, Italy
\\
$^2$ Istituto Nazionale di Fisica Nucleare (INFN), Sezione di Milano,
\\
$^3$ Istituto Nazionale di Alta Matematica (INDAM-GNAMPA)}

\maketitle

\begin{abstract}
The information-theoretic formulation of quantum measurement uncertainty relations (MURs), based on the notion of relative entropy between measurement probabilities, is extended to the set of all the spin components for a generic spin $s$. For an approximate measurement of a spin vector, which gives approximate joint measurements of the spin components, we define the device information loss as the maximum loss of information per observable occurring in approximating the ideal incompatible components with the joint measurement at hand.
By optimizing on the measuring device, we define  the notion of minimum information loss. By using these notions, we show how to give a significant formulation of state independent MURs in the case of infinitely many target observables. The same construction works as well for finitely many observables, and we study the related MURs for two and three orthogonal spin components. The minimum information loss plays also the role of measure of incompatibility and in this respect it allows us to compare quantitatively the incompatibility of various sets of spin observables, with different number of involved components and different values of $s$.

\end{abstract}


\section{Introduction}\label{sec:intro}
In the last twenty years the idea of quantum uncertainty relations has been deeply developed and formalized by introducing different related notions. Measurement uncertainty relations (MURs) for joint measurements quantify to which extent one can approximate a set of measurements of incompatible observables by means of a single joint measurement
\cite{Oza04,Jost+19,BusH08,BusLW14,BLW14b,DamSW15,WernF19,ColesBTW17,BullB18,HeiMZ16,BLPY,RS19,DSABH19}. On the other side, MURs of noise/disturbance type quantify the total uncertainty generated by an approximate measurement of a first observable disturbing the measurement of a second one \cite{BusHOW,CF15,AbbB16,WatN17,BLW14b,WernF19,ColesBTW17,BLPY}. Finally, one speaks of preparation uncertainty relations (PURs) when some lower bound is given on the ``spreads'' of the distributions of some observables measured in the same state \cite{MaaU88,AAHB16,DamSW15,WernF19,ColesBTW17,WehW10,RMM17,HeiMZ16,BLPY,Hol11,GMSS18,KettG19}.  An important point in MURs for joint measurements and PURs is to arrive to formulate them for more than two observables \cite{Jost+19,RMM17,DamSW15,AAHB16,ColesBTW17,WehW10,KettG19,GMSS18}. Various approaches have been proposed to quantify the ``errors'' involved in uncertainty relations, such as variances \cite{Hol11,GMSS18}, distances for probability measures \cite{BLW14b,BusLW14,DamSW15,BullB18,BLPY,RS19}, entropies \cite{MaaU88,ColesBTW17,WehW10}, conditional entropies \cite{BusHOW,CF15,AbbB16}\ldots

In this work, our aim is to develop entropic MURs for all the infinite components of a spin $s$ in the case of an approximate measurement of the full spin vector.   The idea of formulating MURs for all the components of a generic spin $s$ was introduced in \cite{DamSW15}: the measurement of a spin vector is seen as an approximate joint measurement of its infinite components and the aim is to have a quantitative bound on the accuracy with which all these observables can be jointly approximated by such a device.  In \cite{DamSW15}  the approximation error is quantified by Wasserstein distances between target and approximating distributions.
Our approach instead is to see a measurement approximation as a loss of information and to quantify it by the use of the relative entropy \cite{BGT17,BGT18,BarG18}. In information theory, the relative entropy is the notion which allows to quantify the loss of information due to the use of an approximate probability distribution instead of the true distribution. This quantification is independent of a dilation of the measurement units and of a reordering of the possible values. In this context it is possible to arrive to MURs for any set of observables and to quantify their amount of incompatibility.

In \cite{BGT18} we succeeded in formulating state independent MURs for any set of $n$ general observables taking a finite number of possible values. The lower bound appearing in these MURs was named
\emph{entropic incompatibility degree}, and it was shown to play the role of an entropy-based measure of incompatibility.
The generalization to position and momentum was given in \cite{BGT17}. However, the formulation given in these two articles does not extend to infinitely many observables. In \cite{BarG18} we treated the case of all the infinite components of a spin 1/2 system, by an approach based on a mean on the directions. However, this approach cannot be extended to sets of observables for which a natural mean does not exist, and, in any case,
it is very difficult to apply it to higher spins.

In this article we show how to quantify the ``inaccuracy'' in an approximate measurement of the full spin vector, for any value of $s$, by introducing the notion of \emph{device information loss} (Sect.\ \ref{sec:dil}). Then, by optimizing on the measuring apparatus, we define the \emph{minimum information loss} (Sect.\ \ref{sec:mil}), by which the entropic MURs for a spin vector can be expressed, in a state independent form (Sect.\ \ref{sec:MUR}). A key point in the formulation of the MURs is the characterization of the class of approximate joint measurements of all the components of the spin vector (Sect.\ \ref{sec:Msf}). The main difference between the present approach and the one introduced in \cite{BGT18} is that now our focus is on the worst loss of information per observable, while previously it was on the total loss of information.

An important point is that the construction we propose for the spin case allows to formulate MURs also for finite and infinite sets of target observables on the same footing, always in a way that ensures independence from the measurement units, as invariant information theoretical quantities are involved. As a byproduct, this approach will produce also a ``normalized quantity of incompatibility'' (the minimum information loss) for different choices of the target observables; this index  can be used to compare sets of different numbers of observables from the point of view of incompatibility. So, after the construction of MURs for all the spin components in a measurement of the full spin vector, we study also the case of an approximate joint measurement of only 2 or 3 orthogonal spin components and show how the minimum information loss allows the quantitative comparison of the various cases (different numbers of components, different values of $s$). As
already stressed in  \cite{DamSW15}, a joint measurement of three orthogonal components is not equivalent to a joint measurement of all the components, in arbitrary directions, and only the case of infinite components respect the rotation symmetry of an angular momentum. So, it is meaningful to enlighten the differences between the case of the spin components in all directions and the case of orthogonal components.

\paragraph{Scheme of the article.}
In Section \ref{sec:joint_M} we present the approximate joint measurements of all the spin components that we are going to analyze.
These are based on approximate measurements of a spin vector, that is generalized observables on the sphere (Sect.\ \ref{sec:Msf}):
given a \emph{positive operator valued measure} (POVM) on the sphere, we process it into an approximate joint measurement of all the spin components by a projection and discretization procedure of its output (Sect.\ \ref{sec:post-p}).
After a general analysis of the rotational covariant approximate measurements of a spin $s$, more explicit results are given for small spins in Section \ref{sec:smalls}.
In Section \ref{sec:newind} we introduce the \emph{minimum information loss} associated to any approximate measurement of a spin vector.
Such a quantity is the lower bound in the state independent MURs for all the spin components, formulated in Remarks \ref{MURsfirst} and \ref{rem:genMUR}.
We also show that the information loss is minimized in the family of rotational covariant POVMs on the sphere.
In Section \ref{il+nv} we show the connections between our entropic quantity and the incompatibility measures based on generalized noisy versions of the target observables.
The numerical values of the minimum information loss are computed in Sect.\ \ref{sec:ind12} for $s=1/2$, in  Sect.\ \ref{sec:ind1} for $s=1$ and in Sect.\ \ref{sec:ind32} for $s=3/2$. In Section \ref{sec:ind12}  we present  also  a state dependent form of MURs in the special case $s=1/2$.
The MURs for two and three orthogonal components and the corresponding bounds for these cases are introduced in Section \ref{sec:2+3ort}. We show also that the minimum information loss has the role of figure of merit to quantify the incompatibility.
The ordering from the least incompatible set to the more incompatible one is given in Section \ref{sec:order}, for different number of spin components (including the case of infinite components) and different spin values $s$.
Section \ref{sec:concl} presents conclusions and outlooks.

\section{Approximate  joint  measurements  of all  spin  components}\label{sec:joint_M}

In this section we introduce the general notations we shall use, our target observables (the set of all spin components) and the class of their approximating joint measurements.

We fix a Cartesian system $x,\,y,\,z$ determined by the orthogonal unit vectors $\boldsymbol{i},\, \boldsymbol{j},\,\boldsymbol{k}$. Let
$S_x\equiv S_1$, $ S_y\equiv S_2$, $ S_z\equiv S_3$ be an irreducible representation of the commutation relations $[S_x,S_y]=\rmi S_z$ (and cyclic relations) in the \emph{Hilbert space} $\Hscr=\Cbb^{2s+1}$,  so that $S_x^2+S_y^2+S_z^2=s(s+1)\openone$, \ $s=1/2,\,1,\,3/2,\ldots$. The corresponding \emph{state space} (the space of all the statistical operators on $\Hscr$) will be denoted by $\Sscr_s$.  In particular, in some discussions, we shall need the  \emph{maximally mixed state}, given by
\begin{equation}\label{mms0}
\rho_{0}=\frac\openone{2s+1}.
\end{equation}

\subsection{Target observables} \label{sec:jmeas}

We denote by $\Xo\equiv \Xo_1$, $ \Yo\equiv \Xo_2$, $\Zo\equiv \Xo_3$ the projection valued measures associated with the self-adjoint operators $S_x$, $S_y$, $S_z$ (respectively) and by $\Xscr$ the set of possible eigenvalues $m$:
\begin{equation}\label{Xscr}
m\in\Xscr:=\{ -s,-s+1,\ldots, s-1,s\}.
\end{equation}
More in general, for a direction $\boldsymbol{n}$ ($\boldsymbol{n}\in \Rbb^3$, $\abs{\boldsymbol{n}}=1$),  we denote by $\Ao_{\boldsymbol{n}}(m)$  the eigen-projections of the spin component in the direction $\boldsymbol{n}$: \ $\boldsymbol{n}\cdot \boldsymbol{S}= \sum_{m\in \Xscr}m \Ao_{\boldsymbol{n}}(m)$. \
As usual we shall identify $\boldsymbol{n}\cdot \boldsymbol{S}$ and $\Ao_{\bi n}$ by calling both them ``spin component''.

\paragraph{Target observables.} The set of observables which we are going to approximate by joint measurements (the \emph{reference or target observables}) consists of all the spin components (the full spin vector):
\begin{equation}\label{def:Ainfty}
\Acal_\infty:= \left\{\Ao_{\bi n}: \bi n\in \Rbb^3, \; \abs {\bi n}=1\right\}.
\end{equation}

Let us introduce now the usual polar angles $\theta ,\, \phi$  in the fixed reference system and denote by $\bi n(\theta,\phi)$ the unit vector in the direction determined by the polar angles $\theta$ and $ \phi$:
\begin{equation}\label{nTP}
\theta\in [0,\pi], \quad \phi\in [0,2\pi), \qquad \bi n(\theta,\phi)=
\begin{pmatrix} \sin \theta \cos \phi \cr \sin \theta \sin \phi \cr \cos \theta\end{pmatrix}.
\end{equation}

In the following we shall need  the rotation operator
\begin{equation}\label{Sphi}
V(\theta,\phi):=\exp \left\{-\rmi \theta S_\phi\right\},\qquad S_\phi:=S_y\cos \phi-S_x\sin \phi=\rme^{-\rmi\phi S_z}S_y\rme^{\rmi\phi S_z},
\end{equation}
corresponding to a counterclockwise rotation of an angle $\theta$ around the unit vector $\bi n\big(\pi/2, \phi+\pi/2\big)$, see \ref{app:symm}.
Such a rotation  brings the $\bi k $ axis to the $\bi n(\theta,\phi)$ one, so that
\begin{equation}\label{UV}
V(\theta, \phi)S_z V(\theta, \phi)^\dagger=\bi n(\theta,\phi)\cdot \bi S,
\end{equation}
\begin{equation}\label{AV}
V(\theta, \phi)\Zo(m) V(\theta, \phi)^\dagger=\Ao_{\bi n(\theta,\phi)}(m),\qquad m\in\Xscr.
\end{equation}
Finally, the spin components enjoy the covariance property
\begin{equation}\label{covA}
U(R)\Ao_{\bi n}(m)U(R)^\dagger=\Ao_{R\bi n}(m),
\end{equation}
where $U(R)$ is the (projective) representation of $SO_3$ introduced in \ref{app:symm}.

\subsection{Approximate joint measurements}\label{sec:Msf}

We are interested in a measurement of a spin vector, which can be only an approximate measurement otherwise it would be a joint measurement of its components which are all incompatible.
Then an approximate measurement of a spin vector will be seen as an approximate joint measurements of its infinite components.
In some sense, this is even an equivalence if one follows the idea of \cite[Sect.\ 4.1]{DamSW15} that a joint measurement of all components of a vector is a positive operator value measure (POVM) whose output is a vector. We shall come back on this point in Remark \ref{rem:+n-n} and in Section \ref{sec:concl}. For a presentation of POVMs, called also \emph{resolutions of the identity}, see \cite[Sects.\ 4.6, 9.3]{BLPY} and   \cite[Sect.\ 2.2]{Hol11}. We shall denote by $\Mscr(\Yscr)$ the set of all the POVMs with value space $\Yscr$; for instance, we have $\Ao_{\bi n}\in \Mscr(\Xscr)$. The distribution of an observable $\Ao$ in a state $\rho$ will be denoted by $\Ao^\rho$.

The first step is to introduce the set of the approximate joint measurements of the spin vector. As formally the length of a spin is constant, we normalize it to 1 and we consider POVMs on the unit sphere $\mathbb S_2$ in $\Rbb^3$,
\begin{equation}\label{2sph}
\mathbb S_2=\left\{\bxi\in\mathbb{R}^3, \ |\bxi|=1\right\}.
\end{equation}
We denote by $\tilde\Fscr(\mathbb S_2)$ the set of all the POVM on $\mathbb S_2$.

The second step will be to approximate the target observables $\Ao_{\bi n}$ with compatible observables $\Mo_{\bi n}$ that share the same output space $\Xscr$ as $\Ao_{\bi n}$; this will be done in Section \ref{sec:post-p} by processing the output of a POVM on the sphere.

On the physical ground (\cite[Chapt.\ 4]{Hol11}, \cite[Sect.\ 4.4]{DamSW15}), an essential physical property of a measurement of an angular momentum vector is its covariance under the rotation group. Moreover, also when any POVM on the sphere is considered to model a possible measurement of an angular momentum, even if it is not rotational covariant,
one could expect that covariance emerges naturally from any reasonable optimality requirement; it happens in \cite{DamSW15} and the present paper does not make an exception.
Of course, the special properties of rotational covariant POVMs on $\mathbb S_2$ will be the basis of some of our results.
So, here we introduce the covariant POVMs on the sphere and give their properties.

\begin{remark}\label{rem:cov} We denote by $\Fscr(\mathbb S_2)$ the set of all the rotation covariant POVMs on $\mathbb S_2$. The covariance of a POVM  $\Fo\in \Fscr(\mathbb S_2)$ means that, for any Borel subset of the sphere $B\subset \mathbb S_2$ and any rotation $R\in SO(3)$, we have $U(R)\Fo(B)U(R)^\dagger=\Fo(RB)$, where the representation $U(R)$ is introduced in \ref{app:symm}.
\end{remark}

The structure of the POVMs in $\Fscr(\mathbb S_2)$ has been completely characterized in \cite[Sect.\ 4.10]{Hol11}, \cite[p.\ 24]{DamSW15}; any covariant POVM on $\mathbb S_2$ can be expressed as
\begin{equation}\label{mixt}\begin{split}
\Fo_{\blambda}(\rmd \theta \rmd \phi)&=\sum_{\ell=-s}^{+s}\lambda_{\ell}\Fo_\ell(\rmd \theta \rmd \phi),\qquad   \lambda_{\ell}\geq 0,\quad  \sum_{\ell=-s}^{+s}\lambda_{\ell}=1, \quad  \blambda=\{\lambda_\ell\}_{\ell\in \Xscr},
\\
\Fo_\ell(\rmd \theta \rmd \phi)&=
\left(2s+1\right) \Ao_{\bi n(\theta,\phi)}(\ell)\, \frac{\sin\theta\rmd\theta\rmd \phi }{4\pi}.\end{split}
\end{equation}
In particular, the normalization of the measure $\Fo_{\blambda}$ for any choice of the $\lambda$'s implies the normalization of the measures $\Fo_\ell$, which means
\begin{equation}\label{normF}
\int_{\theta\in[0,\pi]}\int_{\phi\in[0,2\pi)}\Fo_\ell(\rmd \theta \rmd \phi)=\openone, \qquad \forall \ell\in \Xscr.
\end{equation}
Let us note that the choice of the $z$-axis is arbitrary.

\begin{remark}[Uniform distribution]\label{uniform}
\begin{enumerate}
\item When
$\lambda_\ell=\lambda^0_\ell\equiv 1/(2s+1)$, $\forall \ell$, \eqref{AV} and \eqref{mixt} imply  that $\Fo_{\blambda^0}(\rmd \theta \rmd \phi)$ is the uniform distribution on the sphere: $\Fo_{\blambda^0}(\rmd \theta \rmd \phi)=\openone\,\frac{\sin\theta}{4\pi}\,\rmd\theta\rmd \phi $.

\item Similarly, for any choice of the parameters $\lambda_m$ we get the uniform distribution on the maximally mixed state \eqref{mms0}:
$\Fo^{\rho_{0}}_{\blambda}(\rmd \theta \rmd \phi)=\frac{\sin\theta }{4\pi}\,\rmd\theta\rmd \phi$.
\end{enumerate}
\end{remark}

\subsubsection{Post-processing.}\label{sec:post-p}

By a natural post-processing procedure, we are now able to construct the compatible observables $\Mo_{\bi n}$ on $\Xscr$, approximating the spin components $\Ao_{\bi n}$. Let $\bxi$ be the result obtained from a measurement on the system of $\Fo\in \tilde\Fscr(\mathbb S_2)$. Being $\bxi$ the observed value, for every direction $\bi n$ we want a value for the ideal spin component $\bi n\cdot \bi S$, obtained by a suitable discretization of $ \bi n \cdot\bxi $. This discretization could be based on different criteria, such as angles of the same amplitude, or  projections on $\bi n$ of the same length. In order to have a sufficiently large class of approximate measurements, we do not ask for such a restrictions; we ask only to have symmetry with respect to positive and negative values, so that we can identify $\bi n\cdot \bi S$ with $-\bi n\cdot \bi S$ up to a change of sign in the output value $m$.

Let us consider a set of angles dividing the interval $[0,\pi]$ into $2s+1$ pieces, symmetrically placed with respect to $\pi/2$:
\begin{equation}\label{thetas}
\btheta= \{\theta_0, \theta_1,\ldots,\theta_{2s+1}\}, \qquad 0=\theta_0<\theta_1<\cdots<\theta_{2s+1}=\pi, \qquad \theta_{2s+1-k}=\pi-\theta_k.
\end{equation}
Let $\bxi$ be the result of the measurement $\Fo$ and $\bi n$ be a generic direction forming an angle $\alpha $ with $\bxi$. If we find $\alpha\in [\theta_{s-m},\theta_{s-m+1})$ for $m=s, \ldots, -s+1$, or $\alpha\in [\theta_{2s},\pi]$ for $m=-s$, we attribute the value $m\in\Xscr$ to the spin component in direction $\bi n$.

In other terms, let $C_{\bi n}(m)$, $m\in \Xscr$, be the $2s+1$ parts of the sphere obtained by using this discretization procedure around $\bi n$; by construction we have
\begin{equation}\label{RCn}
RC_{\bi n}(m)=C_{R\bi n}(m), \qquad \forall R\in SO(3),
\end{equation}
\begin{equation}\label{C+-n}
C_{-\bi n}(m)=C_{\bi n}(-m).
\end{equation}
For any choice of a finite number of directions $\bi n_1,\ldots, \bi n_k$, the approximate joint measurement of the spin components in that directions is represented by
\begin{equation}\label{gen_marg}
\Mo_{\Fo,[\bi n_1,\bi n_2,\ldots,\bi n_k]}(m_1,m_2,\ldots,m_k)=\Fo\left(\bigcap_{i=1}^kC_{\bi n_i}(m_i)\right).
\end{equation}
This expression defines a POVM belonging to $\Mscr(\Xscr^k)$.

\begin{remark}\label{rem:S2+comp}
By the construction we have followed, the  POVMs \eqref{gen_marg}  enjoy many properties; the most relevant properties are the following ones.
\begin{enumerate}
\item When $k$ and $\bi n_1,\ldots,\bi n_k$ vary, the POVMs \eqref{gen_marg} are all compatible, because they are obtained by classical post-processing from a unique measure $\Fo$.
\item By the fact that we have a measure on the space of the directions (the set $\mathbb S_2$) and that the post-processing is described by the intersections in \eqref{gen_marg}, the introduced POVMs are invariant under any permutation of the couples $(\bi n_1,m_1),\ldots , (\bi n_k,m_k)$.
\item Again by the structure \eqref{gen_marg}, the introduced POVMs vanish any time the corresponding intersection among the sets $C_{\bi n_i}(m_i)$ is void.
\item  Equation \eqref{C+-n} implies also the symmetry property
\begin{equation}\label{--}
\Mo_{\Fo,[-\bi n_1,\bi n_2,\ldots,\bi n_k]}(m_1,m_2,\ldots,m_k)=\Mo_{\Fo,[\bi n_1,\bi n_2,\ldots,\bi n_k]}(-m_1,m_2,\ldots,m_k).
\end{equation}
\end{enumerate}
\end{remark}

The set of all these compatible POVMs implicitly defines a measure $\Mo_{\Fo}$ on $\Xscr^{\mathbb S_2}$ for all the spin components; then, the measures \eqref{gen_marg} are $k$-dimensional marginals of $\Mo_{\Fo}$. We denote by  $\Mscr_\infty$ the class of POVMs we get by this procedure: $\Fo\in \tilde\Fscr(\mathbb S_2)$ followed by the post-processing described above.

\begin{remark}\label{rem:+n-n}
Note that, just because of properties (ii)--(iv), $\Mscr_\infty$ is not the class of all the POVM's on $\Xscr^{\mathbb S_2}$ (the class of all the approximate joint measurements of all the spin components). Indeed, this larger class contains also POVMs that do not even enjoy  the natural consistency property
\begin{equation}\label{pm12}
\Mo_{[\bi n,-\bi n]}(m_1,m_2)=0, \qquad \forall  m_2 \neq -m_1
\end{equation}
(recall that $\Ao_{\bi n}(m)=\Ao_{-\bi n}(-m)$). For measures in $\Mscr_\infty$ this property follows from point (iii) in Remark \ref{rem:S2+comp}; indeed, by \eqref{C+-n} we have $C_{\bi n}(m_1)\cap C_{-\bi n}(m_2)= C_{\bi n}(m_1)\cap C_{\bi n}(-m_2)=\emptyset$ if $m_2\neq -m_1$.

Exactly for this reason, here we follow \cite{DamSW15} in starting from measures on the sphere, and we study only POVMs belonging to $\Mscr_\infty$.
\end{remark}

Inside $\Mscr_\infty$ we consider the subclass $\Mscr(\Acal_\infty)$ consisting of all the POVMs $\Mo_{\blambda}\equiv\Mo_{\Fo_{\blambda}}$ obtained by starting from the covariant POVMs $\Fo_{\blambda}$ \eqref{mixt};
we consider such POVMs $\Mo_{\blambda}\in \Mscr(\Acal_\infty)$ as the physically sensible approximate joint measurements of all the spin components $\Acal_\infty$; as a matter of fact, we will prove that they allow to minimize the information lost in the approximation.

\begin{remark}\label{rem:cov+}
By  the covariance of $\Fo_{\blambda}$, the POVMs in $\Mscr(\Acal_\infty)$ enjoy the symmetry property
\begin{equation}\label{symm_marg}
U(R)\Mo_{\blambda,[\bi n_1,\ldots,\bi n_k]}(m_1,\ldots,m_k)U(R)^\dagger =\Mo_{\blambda,[R\bi n_1,\ldots,R\bi n_k]}(m_1,\ldots,m_k).
\end{equation}
\end{remark}

\begin{remark}\label{free_p}
The measure $\Mo_{\blambda}$ depends on $2s+\lfloor s\rfloor $ free parameters: $2s$ parameters from the $\lambda$'s and $\lfloor s\rfloor $ from the angles $\btheta$; $\lfloor s\rfloor $ is the integer part of $s$.
\end{remark}

\subsubsection{The structure of the covariant approximating spin components.}
We study now the structure of the covariant POVMs in $\Mscr(\Acal_\infty)$.
The univariate marginal $\Mo_{\blambda,[\bi n]}$ represents the admissible approximation of $\Ao_{\bi n}$ and its expression turns out to be
\begin{eqnarray}\label{Mn}
\Mo_{\blambda,[\bi n(\theta,\phi)]}(m)=\Fo_{\blambda}\big(C_{\bi n(\theta,\phi)}(m)\big)=V(\theta,\phi)\Mo_{\blambda,[\bi k]}(m)V(\theta,\phi)^\dagger,
\\
\Mo_{\blambda,[\bi k]}(m)=\int_{\theta\in[\theta_{s-m},\,\theta_{s-m+1})}\int_{\phi\in[0,{2\pi}]}\Fo_{\blambda}(\rmd \theta \rmd \phi).
\label{Mk}\end{eqnarray}
The compatible univariate POVMs $\Mo_{\blambda,[\bi n]}$ will be central in our formulation of the MURs and we shall call them ``approximate spin components''.

\begin{remark}
From \eqref{mixt} we see that $\Fo_{\blambda}(\rmd \theta \rmd \phi)$ is a mixture of the POVMs $\Fo_{\ell}(\rmd \theta \rmd \phi)$; similarly, each  $\Mo_{\blambda,[\bi k]}$ is a mixture, given by
\begin{eqnarray}\label{discr_s1}
\Mo_{\blambda,[\bi k]}(m)=\sum_{\ell=-s}^{+s}\blambda_{\ell}\Mo_{\ell,[\bi k]}(m),
\\
\Mo_{\ell,[\bi k]}(m)=\left(2s+1\right) \int_{\theta_{s-m}}^{\theta_{s-m+1}}\rmd \theta\, \frac {\sin\theta}{4\pi}\int_0^{2\pi}\rmd \phi \,\Ao_{\bi n(\theta,\phi)}(\ell). \label{discr_s2}
\end{eqnarray}
In the same way, we have
\begin{equation}\label{discr_s3}
\Mo_{\blambda,[\bi n]}(m)=\sum_{\ell=-s}^{+s}\lambda_{\ell}\Mo_{\ell,[\bi n]}(m), \qquad \Mo_{\ell,[\bi n(\theta,\phi)]}(m)=V(\theta,\phi)\Mo_{\ell,[\bi k]}(m)V(\theta,\phi)^\dagger.
\end{equation}
\end{remark}

In order to study the MURs for spin observables (Sect.\ \ref{sec:newind}), we need a more explicit form for $\Mo_{\blambda,[\bi n]}(m)$, for which the following probabilities are needed.
\begin{definition}[$q$-coefficients]\label{def:q-coeff}
We define
\begin{equation}\label{def:q}
q_{\bi\theta}(m|\ell,h):=\Mo_{\ell,[\bi k]}^{\rho_{h}}(m)=\Tr\left\{\rho_h \Mo_{\ell,[\bi k]}(m)\right\}, \qquad  \rho_{h}:=\Zo(h),
\end{equation}
which is the probability of getting the result $m$ in a measurement of $\Mo_{\ell,[\bi k]}$ when the system is in the eigen-state $\rho_{h}$ of $S_z$. The vector $\btheta$ is the set of the discretization angles \eqref{thetas}, defining $\Mo_{\ell,[\bi k]}$ by \eqref{discr_s2}.
\end{definition}

As stated by the following theorem, the $q$-coefficients involve the \emph{Wigner small-$d$-matrix} \cite[Sect.\ 3.6]{BieL81}, defined by
\begin{equation}\label{Wmatrix}
d^{(s)}_{\ell,h}(\theta):= {}_z\langle \ell|\rme^{-\rmi \theta S_y}|h\rangle_z,\qquad \ell, h\in \Xscr, \qquad \theta\in[0,\pi],
\end{equation}
where $|m\rangle_z$, $m\in\Xscr$, is the normalized eigen-vector of $S_z$ of eigen-value $m$.

\begin{theorem}\label{prop:q}
Each admissible approximate measurement of $\bi n\cdot \bi S$ \eqref{discr_s1}  is diagonal in the basis of the eigen-vectors of $\bi n\cdot \bi S$; indeed, the approximate spin components \eqref{discr_s3} have the form
\begin{equation}\label{margq2}
\Mo_{\ell,[\bi n]}(m)=\sum_{h=-s}^s q_{\bi\theta}(m|\ell,h)\Ao_{\bi n}(h), \qquad \Mo_{\blambda,[\bi n]}(m)=\sum_{\ell,h=-s}^s q_{\btheta}(m|\ell,h)\lambda_{\ell}\Ao_{\bi n}(h),
\end{equation}
where the $q$-coefficients \eqref{def:q} appear. Moreover, these coefficients turn out to be given by
\begin{equation}\label{q+d^2}
q_{\btheta}(m|\ell,h)=\left(s+\frac12\right) \int_{\theta_{s-m}}^{\theta_{s-m+1}}\rmd \theta\,\sin\theta \abs{d^{(s)}_{\ell,h}(\theta)}^2,
\end{equation}
where $d^{(s)}_{\ell,h}(\theta)$ is the Wigner small-$d$-matrix defined in \eqref{Wmatrix}.

Finally, the following properties hold: \ $\forall \,m,\,\ell, \,h\in\Xscr$,
\begin{equation}\label{qpos}
 q_{\btheta}(m|\ell,h)>0, \qquad q_{\btheta}(m|\ell,h)=q_{\btheta}(-m|\ell,-h),
\end{equation}
\begin{equation}\label{qsymm}
 q_{\btheta}(m|\ell,h)=q_{\btheta}(m|h,\ell)=q_{\btheta}(m|-\ell,-h),
\end{equation}
\begin{equation}\label{sumrule2}
\sum_{\ell=-s}^sq_{\btheta}(m|\ell,h)= \sum_{h=-s}^sq_{\btheta}(m|\ell,h)= \left(s+\frac 12\right)\left(\cos\theta_{s-m}-\cos\theta_{s-m+1}\right).
\end{equation}
\end{theorem}

\begin{proof} By using the expressions \eqref{discr_s2} and \eqref{AV} inside the probabilities \eqref{def:q} we get
\[ 
\Mo_{\ell,[\bi k]}^{\rho_{h}}(m)=\left(2s+1\right) \int_{\theta_{s-m}}^{\theta_{s-m+1}}\rmd \theta\, \frac {\sin\theta}{4\pi}\int_0^{2\pi}\rmd \phi \Tr\left\{\Zo(h)V(\theta,\phi)\Zo(\ell)V(\theta,\phi)^\dagger \right\}.
\]
By inserting the  decomposition \eqref{Vdecomp} of $V(\theta,\phi)$, we have that the dependence on $\phi $ disappears and \eqref{q+d^2} is obtained.

The structure of the integral in  $\phi$ in the right hand side of \eqref{discr_s2} implies that $\Mo_{\ell,[\bi k]}(m)$ commutes with  $S_z$ and by the irreducibility of the spin representation it is a linear combination of the projections $\Zo(h)$; by the previous result the coefficients in this expansion are the $q$'s and we get \ $\Mo_{\ell,[\bi k]}(m)=\sum_{h=-s}^s q_{\btheta}(m|\ell,h)\Zo(h)$. \
By \eqref{discr_s3} this proves \eqref{margq2}.

As recalled in \ref{W-d-m}, $\abs{d^{(s)}_{\ell,h}(\theta)}^2$ is a polynomial in $\cos \theta$. As we asked $\theta_{s-m}<\theta_{s-m+1}$, the integral of this polynomial in \eqref{q+d^2} can vanish only if $\abs{d^{(s)}_{\ell,h}(\theta)}^2=0$ for all $ \theta$, but this is impossible because we have
\[
1=\sum_m  q_{\btheta}(m|\ell,h)=\left(s+\frac12\right) \int_{0}^{\pi}\rmd \theta\,\sin\theta \abs{d^{(s)}_{\ell,h}(\theta)}^2,
\]
which follows from \eqref{q+d^2} and the fact that $q_{\btheta}(\bullet|\ell,h)$ is a probability.  Therefore the  strict positivity in \eqref{qpos} holds. The second property in \eqref{qpos} follows from \eqref{pi-m} and the symmetry of the angles in the discretization \eqref{thetas}.

Properties \eqref{qsymm} follow immediately from the definition \eqref{def:q} and the symmetries \eqref{dsymm}.

The sum rules \eqref{sumrule2} follow from the property \eqref{dsum}.
\end{proof}

By \eqref{margq2}, the distribution of an approximate spin component $\Mo_{\blambda,[\bi n]}$ in a state $\rho$ is given by the double mixture
\begin{equation}\label{2mixt}
\Mo_{\blambda,[\bi n]}^\rho(m)=\sum_{\ell,h=-s}^s q_{\btheta}(m|\ell,h)\lambda_{\ell}\Ao_{\bi n}^\rho(h).
\end{equation}

\subsubsection{Noise and compatibility.}\label{noise+C}
Definition \ref{def:q-coeff} says that the $q$-coefficients are probabilities with respect to $m$; then, the quantities
$q_{\bi\theta}(\bullet|\ell,\bullet)$ and $\sum_{\ell=-s}^s q_{\btheta}(\bullet|\ell,\bullet)\lambda_{\ell}$ are transition matrices, independent of the system state $\rho$. Then, equations \eqref{margq2} and \eqref{2mixt} can be interpreted by saying that, given the direction $\bi n$, each covariant approximating spin component $\Mo_{\ell,[\bi n]}$ or $\Mo_{\blambda,[\bi n]}$ could be obtained by measuring exactly the target observable $\Ao_{\bi n}$ and then by perturbing the result  with some classical noise through a one-step stochastic evolution given by one of the transition matrices just introduced.
As we have seen in Remark \ref{rem:S2+comp}, the univariate POVMs $\Mo_{\blambda,[\bi n]}$ are all compatible because they are obtained by a classical post-processing from the unique POVM $\Fo_{\blambda}$; the compatibility is not implied by the structure \eqref{margq2} alone. The use of classical transition matrices (Markov kernels) to transform incompatible observables into compatible ones has already been exploited in related problems \cite{BLPY,DamSW15,Haap15}.

A different approach \cite{BLPY,HSTZ14,Haap15,HKR15,DFK19,HeiMZ16} to the construction of compatible observables is to consider noisy versions of the target observables.

\begin{definition}
If $\Do$ is an observable and $\No$ another POVM with the same value space, the mixture
\[
\Do'=\eta \Do+(1-\eta)\No, \qquad \eta\in [0,1],
\]
is said to be a \emph{noisy version} of the observable $\Do$ with noise $\No$ and \emph{visibility} $\eta$.
\end{definition}

Given the target observables $\Do_j$, $j\in I$, and the class of permitted noises, the problem considered in the quoted references is to see how much noise has to be added to the target observables in order to get \textbf{compatible} POVMs of the form $\Do_j'= \eta\Do_j+(1-\eta)\No_j$.
The various approaches in the literature differ for the classes of admissible noises; often only classical noise is considered, i.e.\ $\No_j(\bullet)=p_j(\bullet)\openone$ where $p_j$ is a classical probability, independent of the system state \cite{HeiMZ16,HSTZ14,HKR15}. A review of some choices for the noise classes introduced in the literature is given in \cite{DFK19}; the typical choices are: (a) classical noises, (b) noises represented by compatible POVMs, (c) general POVMs.

The marginals of an approximating joint measurement in $\Mscr(\Acal_\infty)$  can be expressed as noisy versions of the corresponding target observables in a way which will be useful for comparisons, as stated in the following remark.

\begin{remark}
The marginals \eqref{margq2} of the joint measurement $\Mo_{\blambda}$ can be written in the form
\begin{equation}\label{Mo-nv}
\Mo_{\blambda, [\bi n]}(m)=\eta_{\blambda,\,\btheta}\Ao_{\bi n}(m)+\left(1-\eta_{\blambda,\,\btheta}\right)\No^{\blambda,\btheta}_{\bi n}(m),
\end{equation}
\begin{equation}\label{etavis}
\eta_{\blambda,\,\btheta}=\min_{m\in\Xscr}\sum_{\ell \in\Xscr}q_{\btheta}(m|\ell,m)\lambda_\ell, \qquad 0<\eta_{\blambda,\,\btheta}<1,
\end{equation}
\begin{multline}\label{Nstr}
\No^{\blambda,\btheta}_{\bi n}(m)=\frac 1{1-\eta_{\blambda,\,\btheta}}\biggl\{\biggl[\sum_{\ell\in\Xscr} q_{\btheta}(m|\ell,m)\lambda_{\ell}-\min_{m'\in\Xscr}\sum_{\ell \in\Xscr}q_{\btheta}(m'|\ell,m')\lambda_\ell\biggr]\Ao_{\bi n}(m)\\ {}+\sum_{\ell,h\in\Xscr}\left(1-\delta_{hm}\right) q_{\btheta}(m|\ell,h)\lambda_{\ell}\Ao_{\bi n}(h)\biggr\}.
\end{multline}
It is easy to see that $\No^{\blambda,\btheta}(m)_{\bi n}$ is positive and that $\No^{\blambda,\btheta}_{\bi n}$ is indeed a POVM; then, the proof of the decomposition \eqref{Mo-nv} is trivial. This simple expression is due to the fact that each target POVM $\Ao_{\bi n}$ and its approximating POVM $\Mo_{\blambda, [\bi n]}$ are diagonal on the same basis. Due to covariance, the visibility $\eta_{\blambda,\,\btheta}$ does not depend on $\bi n$. Due to the strict positivity \eqref{qpos} of the $q$-coefficients, the visibility is strictly positive; moreover, it cannot be 1, which is possible only when the target observables are already compatible.
\end{remark}

In expressing $\Mo_{\blambda, [\bi n]}$ as a mixture of $\Ao_{\bi n}$ and some ``noise'', the decomposition is not unique. In writing the decompositions \eqref{Mo-nv} we have decided to have the maximum possible value for the visibility $\eta_{\blambda,\,\btheta}$, without imposing conditions on the class of allowed noises. As we remarked above, the last class of noises discussed in \cite{DFK19} is indeed the one of general POVMs. If the class of noises is restricted, the value of the visibility could diminish, as we can see in the example of $s=1/2$, Sect.\ \ref{joint1/2}.

\subsubsection{Unbiased measurements.}\label{sec:unb}
Sometimes, not only symmetries are used to restrict the class of possible approximate joint measurements of some incompatible target observables.
In \cite{BullB18,YLLOh10,YuOh13} spin measurements with \emph{unbiased} marginals are considered; by this they mean that the outcomes of the measurement are uniformly distributed when the system is in the maximally mixed state.  Note that in the field of inferential statistics this term has a different meaning, cf.\ \cite[Chapt.\ 6]{Hol11}.

By taking into account that our target observables $\Ao_{\bi n}$ are indeed unbiased in this sense, it could be reasonable to ask this restriction also for  the approximating observables. In the case of covariant approximate joint measurements, by \eqref{2mixt} and \eqref{sumrule2}, to ask the uniform distribution $\Mo_{\blambda,[\bi k]}^{\rho_0}(m)=1/(2s+1)$, in the maximally mixed state $\rho_0$ \eqref{mms0}, implies immediately the strong restriction
\begin{equation}\label{Deltacos}
\cos\theta_{k}-\cos\theta_{k+1}=\frac 2{2s+1},\qquad \text{i.e.} \qquad \cos\theta_k=\frac{2s+1-2k}{2s+1}.
\end{equation}
This choice corresponds to discretize $\bi n \cdot \bxi$ by dividing the interval $[-1,1]$ into subintervals of equal length.
By using the minimization of information loss as criterium of goodness, as done in Sect.\ \ref{sec:newind},  the best approximate joint measurement not always satisfies this restriction (see Sections \ref{sec:ind1}, \ref{sec:ind32}) and we do not ask for unbiasedness.
Also in other contexts, biased measurements turned out to be optimal \cite{HKR15}.

\subsection{Covariant approximate joint measurements for spin 1/2, 1, 3/2}\label{sec:smalls}
For small spins we can get explicit results by particularizing the discretization procedure of  Section \ref{sec:post-p} and using the $q$-coefficients computed in \ref{app:qcoeff}.
\subsubsection{Spin 1/2.}\label{joint1/2}
In this case only three angles appear in the post-processing and they are completely determined by \eqref{thetas}: $\theta_0=0$, \ $\theta_1=\pi /2 $, \ $ \theta_{2}=\pi$. \ So, no free parameter is introduced by the discretization of the directions and a single free parameter remains, coming from the $\lambda$'s, see Remark \ref{free_p}. These angles automatically satisfy \eqref{Deltacos} and this means that for $s=1/2$ any observable in $\Mscr(\Acal_\infty)$ is unbiased in the sense of Section \ref{sec:unb}.

The most general expression of the approximate spin components \eqref{discr_s1}, \eqref{discr_s2} has been already obtained in \cite[Sect.\ 5]{BarG18}, but it can be computed also from the explicit form of the $q$-coefficients given in \eqref{q1/2}:
\begin{equation}\label{Mk1/2}
\Mo_{\blambda,[\bi k]}(m)=\frac\openone 2 + \left(\lambda_{1/2}-\frac 12\right) 2m S_z, \qquad \lambda_{1/2}\in[0,1].
\end{equation}
By using  $\Zo(m)=\openone-\Zo(-m)$, we can rewrite \eqref{Mk1/2} as
\begin{multline}\label{kmarg1/2}
\Mo_{\blambda,[\bi k]}(m)=\left(\frac 32 -\lambda_{1/2}\right)\frac\openone 2 + \left(\lambda_{1/2}-\frac 12\right)\Zo(m) \\ {}= \left(\frac 12 +\lambda_{1/2}\right)\frac\openone 2 + \left(\frac 12-\lambda_{1/2}\right)\Zo(-m),
\end{multline}
from which we see that $\Mo_{\blambda,[\bi k]}$ is a noisy version of $\Zo$ with classical noise, only when $\lambda_{1/2} \geq \frac 12$.

By allowing general noises, we have the structure \eqref{Mo-nv}, which is a different decompo\-sition of the approximating measures as noisy versions of the target observables. For $s=1/2$, by particularizing \eqref{etavis} and \eqref{Nstr}, we see that the $\btheta$ dependence disappears and the explicit expressions of visibility and noise become
\begin{equation}
\eta_{\blambda}=\frac 14 +\frac{ \lambda_{1/2}}2, \qquad \No^{\blambda}_{\bi n}(m)=\No_{\bi n}(m)=\Ao_{\bi n}(-m).
\end{equation}
Note that in the decomposition \eqref{Mo-nv} for $s=1/2$ the noises turn out to be projection valued measures and they do not commute for different directions; so, the noises $\No_{\bi n}$, $\bi n\in \mathbb{S}_2$, are incompatible. We have asked the compatibility of the POVMs $\Mo_{\blambda,[\bi n]}$, not of the noises.

For $s=1/2$ the probabilities \eqref{2mixt} can be easily computed. Firstly, any state can be parameterized as
\begin{equation}\label{state1/2}
\rho=\frac 12 \left(\openone +2\bi r\cdot \bi S\right),\qquad r=\abs{\bi r}\leq 1;
\end{equation}
note that $2\bi S$ is the vector of the Pauli matrices. Then, by \eqref{AV} and \eqref{Mk1/2}, we have
\begin{equation}\label{prob1/2}
\Ao_{\bi n}^\rho(m)=\frac 12+m \,\bi n \cdot \bi r, \qquad \Mo_{\blambda,[\bi n]}^\rho(m)=\frac 12+\left(\lambda_{1/2}-\frac 12\right) m\,\bi n \cdot \bi r.
\end{equation}

\subsubsection{Spin 1.}\label{joint1}
The choice of the angles \eqref{thetas} gives
$0=\theta_0<\theta_1 < \theta_2=\pi-\theta_1<\theta_3=\pi$,
and it introduces a single free parameter
\begin{equation} \label{s1_a}
a:=\cos\theta_1, \qquad a\in(0,1).
\end{equation}
Other two free parameters come from the $\lambda$'s, see Remark \ref{free_p}. The $q$-coefficients are computed in
\ref{app:s=1}; then, the approximate spin components \eqref{margq2} take the expressions
\begin{subequations}\label{M11}
\begin{multline}
\Mo_{1,[\bi k]}(\pm 1)=\Mo_{-1,[\bi k]}(\mp 1) 
\\ {}=\left[1-\frac{(1+a)^3}8\right]\Zo(\pm 1)+\frac{2+a}4\,(1-a)^2\,\Zo(0)+\frac{(1-a)^3}8\, \Zo(\mp 1),
\end{multline}
\begin{equation}
\Mo_{1,[\bi k]}(0)=\Mo_{-1,[\bi k]}(0)=\frac a2 \left(3-a^2\right)\Zo(0) + \frac a4\left(3+a^2\right) \left[\Zo(1)+\Zo(-1)\right] ,
\end{equation}
\end{subequations}
\begin{equation}\label{M00}\begin{split}
&\Mo_{0,[\bi k]}(\pm 1)=\frac{2+a}4\,(1-a)^2\left[\Zo(1)+\Zo(-1)\right]+ \frac {1-a^3}2\,\Zo(0),
\\
&\Mo_{0,[\bi k]}(0)=a^3\Zo(0)+\frac{a}2\,(3-a^2)\left[\Zo(1)+\Zo(-1)\right].\end{split}
\end{equation}

To get unbiased marginals, according to \eqref{Deltacos} we would have to take $a=1/3$; as we already wrote in Section \ref{sec:unb} we do not ask for this and we leave free the parameter $a$.

\subsubsection{Spin 3/2.}\label{joint3/2}
For $s=3/2$, the choice of the angles \eqref{thetas} gives
\[ 0=\theta_0<\theta_1 < \theta_2=\frac \pi2< \theta_3=\pi-\theta_1< \theta_4=\pi,
\]
and it introduces a single free parameter: $a:=\cos\theta_1$, \ $a\in(0,1)$.
Other three free parameters come from the $\lambda$'s, see Remark \ref{free_p}. The $q$-coefficients are computed in
\ref{app:spin3/2}; then, the approximate spin components are given by \eqref{margq2}, \eqref{discr_s1} and the probability distribution by \eqref{2mixt} (we to not write explicitly them, because the formulae are very long). To get unbiasedness, according to \eqref{Deltacos} we would have to take $a=1/2$.

\section{Entropic MURs for the set of all the spin components}\label{sec:newind}
A spin vector can not be exactly measured, as its components are incompatible observables  and a joint measurement can only  approximate them. In information theory \cite{Ved02,BA02,CovT06} the relative entropy is the quantity introduced to measure the error done when one uses an approximating probability distribution in place of the true one. Let us stress that the relative entropy is an intrinsic quantity: it is independent of the measure units of the involved observables and from renaming or reordering the possible values. Such a property does not hold for non entropic measures of the error.

In \cite{BGT18} we used as error function the sum of the relative entropies, each one involving a single target observable, because this sum represents the total loss of information; however, this approach can not be extended to infinitely many observables. To overcome this difficulty, instead of the sum, we shall consider the maximum of the relative entropies over all target observables: this maximum represents the loss of information for the worst direction. Then, we consider the worst case also with respect to the system state. Finally, we shall optimize with respect to all approximating joint measurements. This is indeed the procedure used in \cite{DamSW15,Jost+19,BullB18}, apart from the starting point (distances between distributions for them).

\subsection{The device information loss}\label{sec:dil}

Let us recall that $\Acal_\infty$ \eqref{def:Ainfty} is the set of all the spin components (our target observables), that $\Mscr_\infty$ is the class of the approximate joint measurements for all the spin components, and that $\Mscr(\Acal_\infty)$ is the class of the covariant ones, $\Mscr(\Acal_\infty)\subset \Mscr_\infty$ (see Sect.\ \ref{sec:post-p}). If $\Ao_{\bi n}\in \Acal_\infty$ and $\Mo\in \Mscr_\infty$, we denote by  $\Mo_{[\bi n]}$ the univariate marginal of $\Mo$ approximating $\Ao_{\bi n}$ and we call it the approximate spin component. With $\Ao_{\bi n}^\rho$ we denote the distribution of $\Ao_{\bi n}$ in the state $\rho$, and similar notation for the other observables.

To quantify the information loss due to the use of $\Mo^\rho_{[\bi n]}$ in place of the target distribution $\Ao_{\bi n}^\rho$, we take the relative entropy
\begin{equation}\label{relentMon}
S\big(\Ao_{\bi n}^\rho\big\|\Mo_{[\bi n]}^\rho\big)=\sum_{m\in \Xscr}\Ao_{\bi n}^\rho(m)\log \frac{\Ao_{\bi n}^\rho(m)}{ \Mo_{[\bi n]}^\rho(m)}\geq 0, \qquad \Mo\in \Mscr_\infty,
\end{equation}
where the logarithm is with base 2: $\log \equiv \log_2$. Recall that the form $0\log 0$ is taken to be zero and that the relative entropy can be $+\infty$ when the support of the second probability distribution is not contained in the support of the first one. When a covariant measurement is considered, by using the expression of $\Mo^\rho_{\blambda,[\bi n]}$ in terms of the $\lambda$'s and the $q$-coefficients in \eqref{2mixt}, we have
\begin{equation}\label{Skq}
S\big(\Ao^\rho_{\bi n}\|\Mo^\rho_{\blambda,[\bi n]}\big)=\sum_{m\in\Xscr}\Ao^\rho_{\bi n}(m)\log \frac{\Ao^\rho_{\bi n}(m)} {\sum_{\ell,h} q_{\btheta}(m|\ell,h) \lambda_{\ell}\Ao^\rho_{\bi n}(h)}, \qquad \Mo_{\blambda}\in \Mscr(\Acal_\infty).
\end{equation}
As all the $q$-coefficients are strictly positive \eqref{qpos}, the relative entropy \eqref{Skq} is always finite.

The relative entropy \eqref{relentMon} depends on the state and on the choice of the observable (the direction $\bi n$).  To characterize an information loss due only to the measuring device, represented by the multi-observable $\Mo$ approximating all the observables in $\Acal_\infty$, we consider the worst case of \eqref{relentMon} with respect to the system state and the measurement direction. So, we define the \emph{device information loss} by
\begin{equation}\label{sie}
\Delta_s[\Acal_\infty\|\Mo]:=\sup_{\rho\in\Sscr_s,\;\bi n\in\mathbb{S}_2}S\big(\Ao_{\bi n}^\rho\big\|\Mo_{[\bi n]}^\rho\big), \qquad \Mo\in \Mscr_\infty.
\end{equation}

This quantity is the analogue of the entropic divergence introduced in \cite[Definition 2]{BGT18}); to use the worst case on the directions instead of the sum of the relative entropies, as done there, allows to consider also infinitely many target observables. Alternatively, in \cite{BarG18} we started from the mean of the relative entropies made over all the directions, but this approach gives rise to computations intractable outside the case $s=1/2$, and without possible extensions in cases in which an invariant mean does not exist.

\begin{theorem}
\label{prop:infloss}
The device information loss \eqref{sie} is always strictly positive:
\begin{equation}\label{Deltapos}
\Delta_s[\Acal_\infty\|\Mo]>0, \qquad \forall \Mo\in \Mscr_{\infty}.
\end{equation}
Moreover, $\forall \Mo\in \Mscr_{\infty} $ there exists $\hat \Mo\in \Mscr(\Acal_\infty)$ such that
\begin{equation}\label{Deltasymm}
\Delta_s[\Acal_\infty\|\hat\Mo]\leq \Delta_s[\Acal_\infty\|\Mo].
\end{equation}

In the case of a covariant measurement, the double supremum in the definition \eqref{sie} of the device information loss is a maximum, and we have
\begin{equation}\label{ntok}
\Delta_s[\Acal_\infty\|\Mo]=\max_{\rho\in\Sscr_s}S(\Ao_{\bi n}^\rho\|\Mo^\rho_{[\bi n]})
<+\infty, \qquad \forall \bi n, \quad \forall \Mo\in \Mscr(\Acal_\infty).
\end{equation}
Moreover, the maximum over the states is realized in an eigen-projection of the spin component:
\begin{equation}\label{Delta+S}
\Delta_s[\Acal_\infty\|\Mo]= \max_{m\in \Xscr}S(\Ao_{\bi n}^{\rho^{\bi n}_m}\|\Mo^{\rho^{\bi n}_m}_{[\bi n]}), \qquad \rho^{\bi n}_m:= \Ao_{\bi n}(m), \qquad \forall \Mo\in \Mscr(\Acal_\infty).
\end{equation}
Finally, in terms of the $q$-coefficients \eqref{def:q}, the device information loss  \eqref{sie} is given by
\begin{equation}\label{Delta+q}
\Delta_s[\Acal_\infty\|\Mo_{\blambda}] =\log \left(\min_{m\in \Xscr}\sum_{\ell}\lambda_{\ell}q_{\btheta}(m|\ell,m)\right)^{-1}, \qquad \forall \Mo_{\blambda}\in \Mscr(\Acal_\infty).
\end{equation}
\end{theorem}

\begin{proof}
The relative entropy is equal to zero if and only if the two probability distributions coincide; by the incompatibility of the spin observables, the device information loss \eqref{sie} is strictly positive and \eqref{Deltapos} is proved.

To prove \eqref{Deltasymm}, we need the notion of
symmetrized version of a generic POVM on the sphere. The symmetrization $\hat \Fo$ of $\Fo\in \tilde\Fscr(\mathbb S_2)$ is defined by
\begin{equation}\label{ize}
\hat \Fo(B)=\int_0^{2\pi} \rmd \phi \int_0^\pi \rmd \theta \, \frac {\sin \theta}{4\pi} \, V(\theta,\phi) \Fo\big(R_{\bi u(\phi)}(\theta)^{-1} B\big) V(\theta,\phi)^\dagger,
\end{equation}
where the rotation $R_{\bi u(\phi)}(\theta)$ and the corresponding unitary operator $V(\theta,\phi)$ are defined in  equations \eqref{R_u}, \eqref{V=UR}. One can check that the covariance property, given in Remark \ref{rem:cov}, holds for $\hat \Fo$, and that a covariant POVM is left invariant by the transformation \eqref{ize}:
\begin{equation*}
\Fo\in \tilde\Fscr(\mathbb S_2)\Rightarrow \hat \Fo\in \Fscr(\mathbb S_2), \qquad \Fo\in \Fscr(\mathbb S_2)\Rightarrow \hat \Fo=\Fo.
\end{equation*}
From $\hat \Fo$, by the post-processing \eqref{gen_marg}, we construct $\hat\Mo\in \Mscr(\Acal_\infty)$; by the property \eqref{RCn} and the definition \eqref{ize}, we get
\begin{equation}\label{hatMn} 
\hat \Mo_{[\bi n]}(m)=\hat \Fo\big(C_{\bi n}(m)\big)=\int_0^{2\pi} \rmd \phi \int_0^\pi \rmd \theta \, \frac {\sin \theta}{4\pi} \, V(\theta,\phi) \Fo\big( C_{R(\theta,\phi)^{-1}\bi n}(m)\big) V(\theta,\phi)^\dagger.
\end{equation}
By this construction, any $\Mo \in \Mscr_\infty$ is generated by post-processing some $\Fo\in\tilde\Fscr(\mathbb S_2)$; let $\hat \Fo$  be the symmetrization \eqref{ize} of $\Fo$ and let $\hat \Mo\in \Mscr(\Acal_\infty)$ be the corresponding measure obtained by post-processing $\hat \Fo$. Now, we set
\[
\rho(\theta,\phi):=V(\theta,\phi)^\dagger\rho V(\theta,\phi), \qquad R(\theta,\phi):= R_{\bi u(\phi)}(\theta),
\]
where  $R_{\bi u(\phi)}(\theta)$ is the rotation involved in $V(\theta,\phi)$, see \eqref{V=UR}.
By \eqref{hatMn}, \eqref{covA}, and the convexity of the relative entropy, we get
\begin{multline*}
S(\Ao_{\bi n}^\rho\|\hat \Mo^\rho_{[\bi n]})\leq \int_0^{2\pi}\rmd \phi \int_0^\pi \rmd \theta \, \frac {\sin \theta}{4\pi} \, S\Big(\Ao_{\bi n}^\rho\Big\|\Tr\left\{ \rho(\theta,\phi) \Fo\big( C_{R(\theta,\phi)^{-1}\bi n}(m)\big) \right\}\Big)
\\ {}=\int_0^{2\pi}\rmd \phi \int_0^\pi \rmd \theta \, \frac {\sin \theta}{4\pi}\,S\Big(\Ao_{R(\theta,\phi)^{-1}\bi n}^{\rho(\theta,\phi)}\Big\|\Tr\left\{ \rho(\theta,\phi) \Fo\big(C_{R(\theta,\phi)^{-1}\bi n}(m)\big) \right\}\Big).
\end{multline*}
By taking the supremum of the definition \eqref{sie} we get
\begin{multline*}
\Delta_s[\Acal_\infty\|\hat\Mo]\leq \sup_{\rho\in\Sscr_s,\;\bi n\in\mathbb{S}_2}\int_0^{2\pi}\rmd \phi \int_0^\pi \rmd \theta \, \frac {\sin \theta}{4\pi}\,S\Big(\Ao_{R(\theta,\phi)^{-1}\bi n}^{\rho(\theta,\phi)}\Big\|\Tr\left\{ \rho(\theta,\phi) \Fo\big( C_{R(\theta,\phi)^{-1}\bi n}(m)\big) \right\}\Big)
\\  {}\leq \int_0^{2\pi}\rmd \phi \int_0^\pi \rmd \theta \, \frac {\sin \theta}{4\pi}\,\sup_{\rho\in\Sscr_s,\;\bi n\in\mathbb{S}_2}S\Big(\Ao_{\bi n}^{\rho}\Big\|\Tr\left\{ \rho \Fo\big( C_{\bi n}(m)\big) \right\}\Big),
\end{multline*}
and this gives \eqref{Deltasymm}.

Now we take $\Mo\in \Mscr(\Acal_\infty)$. In the double $\sup$ in \eqref{sie} we can execute the supremum over the states first. By covariance, the quantity $\sup_{\rho\in\Sscr_s}S(\Ao_{\bi n}^\rho\|\Mo^\rho_{[\bi n]})$ is independent of $\bi n$ and we obtain
\[
\Delta_s[\Acal_\infty\|\Mo]=\sup_{\rho\in\Sscr_s}S(\Ao_{\bi n}^\rho\|\Mo^\rho_{[\bi n]})=\sup_{\rho\in\Sscr_s}S(\Zo^\rho\|\Mo^\rho_{[\bi k]}).
\]
By convexity, the supremum over the states of the expression \eqref{Skq} is a maximum among the $2s+1$ eigen-states of $S_z$ and we get \eqref{Delta+S}, the equality in \eqref{ntok}, and
\[
\sup_{\rho\in \Sscr_s}S(\Zo^\rho\|\Mo^\rho_{[\bi k]})=\max_{m\in \Xscr}\log \left(\sum_{m'}\lambda_{m'}q(m|m',m)\right)^{-1}.
\]
Then, the device information loss \eqref{sie} can be written in the form \eqref{Delta+q}, which is finite because of the strict positivity \eqref{qpos} of the $q$'s.
\end{proof}

\subsection{The minimum information loss}\label{sec:mil}

By optimizing over the class $\Mscr(\Acal_\infty)$ of the physical approximating measurements we get a lower bound for the device information loss
\begin{equation}\label{def:Is}
I_s[\Acal_\infty\|\Mscr(\Acal_\infty)]:=\inf_{\Mo\in\Mscr(\Acal_\infty)}\Delta_s[\Acal_\infty\|\Mo];
\end{equation}
we call it \emph{minimum information loss}. An analogous quantity can be defined also for the larger class $\Mscr_\infty$:
\begin{equation}\label{def:Is2}
I_s[\Acal_\infty\|\Mscr_\infty]:=\inf_{\Mo\in\Mscr_\infty}\Delta_s[\Acal_\infty\|\Mo].
\end{equation}
The two minimum information losses turn out to be equal, as shown in Theorem \ref{Iprop}.

The quantity $I_s[\Acal_\infty\|\Mscr(\Acal_\infty)]$ has interesting properties; in particular, as shown in Theorem \ref{Iprop}, it is strictly positive. Moreover, in the spin definition given in Section \ref{sec:jmeas} we have used  $\hbar=1$, but \eqref{def:Is} is independent of this choice, because of the invariance properties of the relative entropy. The minimum information loss will appear in the formulations of the MURs (Sect.\ \ref{sec:MUR}) and it can be used as a measure of the incompatibility of the set of the target observables.
The expression \eqref{def:Is} can be elaborated and a more explicit form can be obtained.

\begin{theorem}\label{Iprop}
The two information losses \eqref{def:Is} and \eqref{def:Is2} are equal:
\begin{equation}\label{I=I}
I_s[\Acal_\infty\|\Mscr(\Acal_\infty)]=I_s[\Acal_\infty\|\Mscr_\infty].
\end{equation}
The minimum information loss \eqref{def:Is} can be expressed in terms of the $q$-coefficients \eqref{def:q} as
\begin{equation}\label{I+q}
I_s[\Acal_\infty\|\Mscr(\Acal_\infty)]=\log \left(K_s\right)^{-1}, \qquad K_s:=\sup_{\blambda, \btheta}\min_{m}\sum_{\ell }\lambda_{\ell}q_{ \btheta}(m|\ell,m),
\end{equation}
where $\btheta$ is the set of angles satisfying the discretization conditions \eqref{thetas} and involved in the expression \eqref{q+d^2} of the $q$-coefficients.
Moreover, the following bounds hold:
\begin{equation}\label{I<}
0<I_s[\Acal_\infty\|\Mscr(\Acal_\infty)]\leq \log\left(2s+1\right).
\end{equation}
\end{theorem}

\begin{proof}
Obviously, we have $I_s[\Acal_\infty\|\Mscr(\Acal_\infty)]\geq I_s[\Acal_\infty\|\Mscr_\infty]$, because $\Mscr(\Acal_\infty)\subset \Mscr_\infty$. The opposite inequality is implied by \eqref{Deltasymm}; so,  equality \eqref{I=I} is proved.

To get $I_s[\Acal_\infty\|\Mscr(\Acal_\infty)]$ from \eqref{Delta+q}, one has to minimize over the $\lambda$'s and the discretization angles:
\[
I_s[\Acal_\infty\|\Mscr(\Acal_\infty)]=\inf_{\blambda, \btheta}\log \left(\min_{m}\sum_{\ell}\lambda_{\ell}q_{\bi\theta}(m|\ell,m)\right)^{-1} 
=\log\left(\sup_{\blambda, \btheta}\min_{m}\sum_{\ell}\lambda_{\ell}q_{\bi\theta}(m|\ell,m)\right)^{-1};
\]
this gives \eqref{I+q}.
Then, with the choice $\lambda_{\ell}=1/(2s+1)$ and \eqref{Deltacos} for the angles, we have
\[
K_s\geq \sup_{\btheta}\max_m\sum_{\ell}\frac{q_{\btheta}(m|\ell,m)}{2s+1}=\frac 12\left(\cos\theta_{s-m}-\cos\theta_{s-m+1}\right)=(2s+1)^{-1};
\]
this proves the upper bound in \eqref{I<}.

To prove the first inequality in \eqref{I<} we relay on the results of \cite{BGT18}. The entropic incompatibility degree for two target observables, defined in \cite[(10)]{BGT18}, is strictly positive when the two observables are incompatible \cite[Theor.\ 2, point (v)]{BGT18}. Moreover, the class of the POVMs on $\Xscr^2$, $\Mo\in\Mscr(\Xscr^2)$, is larger than the class of the bivariate marginals of measures in $\Mscr(\Acal_\infty)$. By starting from  two orthogonal spin components, $\Xo,\Yo$, we get
\begin{multline*}
0\overset{(1)}{<} c_{\rm inc}(\Xo,\Yo)\overset{(2)}{=}\inf_{\Mo\in\Mscr(\Xscr^2)}\sup_{\rho\in\Sscr_s}\sum_{i=1}^2S\big( \Xo^\rho_i\|\Mo^\rho_{[i]}\big) \overset{(3)}{\leq} \inf_{\Mo\in\Mscr(\Xscr^2)}\sup_{\rho\in\Sscr_s}2\max_{i=1,2}S\big( \Xo^\rho_i\|\Mo^\rho_{[i]}\big)
\\ {}\overset{(4)}{\leq} 2\inf_{\Mo\in\Mscr(\Acal_\infty)}\sup_{\rho\in\Sscr_s}\max_{i=1,2}S\big( \Xo^\rho_i\|\Mo^\rho_{[i]}\big) \overset{(5)}{\leq} 2\inf_{\Mo\in\Mscr(\Acal_\infty)}\sup_{\rho\in\Sscr_s,\;\bi n\in\Rbb^3,\; \abs {\bi n}=1}S\left(\Ao_{\bi n}^\rho\|\Mo_{[\bi n]}^\rho\right)
\\ {}\overset{(6)}{=}2I_s[\Acal_\infty\|\Mscr(\Acal_\infty)].
\end{multline*}
Here (1) is the result of \cite{BGT18}, (2) is the definition of $c_{\rm inc}$, (3) is because we  substitute the sum with two times the maximum, (4) is because we have restricted the class of approximating joint measurements in the infimum, (5) is because we enlarge the set of directions in the maximum, (6) is by our definition \eqref{sie}, \eqref{def:Is}. This ends the proof of the strict positivity.
\end{proof}

Let us remark that the last part of the proof, proving the strict positivity in \eqref{I<}, works for every class of approximate joint measurements one could use in the infimum, not only for our choices $\Mscr(\Acal_\infty)$ and $\Mscr_\infty$. The point is that every spin component $\Ao_{\bi n}$ has to be approximated by a POVM $\Mo_{[\bi n]}$ on the same output space $\Xscr$ and that the $\Mo_{[\bi n]}$, $\bi n\in \mathbb{S}_2$, must be compatible.

\subsection{Entropic MURs}\label{sec:MUR}

By the strict positivity of the minimum information loss proved in Theorem \ref{Iprop}, the definitions \eqref{def:Is}, \eqref{def:Is2}, and the equality \eqref{I=I}, we get a first formulation of the MURs, in a state independent form, which is analogous to that given in \cite[(11)]{DamSW15}.

\begin{remark}[MURs, first version]\label{MURsfirst} For every approximate joint measurement $\Mo$ of all the spin components, the device information loss \eqref{sie} is greater than a strictly positive lower bound:
\[
\Delta_s[\Acal_\infty\|\Mo] \geq I_s[\Acal_\infty\|\Mscr(\Acal_\infty)]>0, \qquad \forall \Mo\in\Mscr_\infty.
\]
\end{remark}

By the comments above
we have that non trivial entropic MURs can be formulated also if we change the class of approximate joint measurements $\Mscr_\infty$ with some other class; what can change is the value of the (strictly positive) minimum information loss.

\begin{remark}\label{optimalMo}
By the expression \eqref{Delta+S} of the device information loss, we can write  \eqref{I+q} as
\begin{equation}\label{I+S}
I_s[\Acal_\infty\|\Mscr(\Acal_\infty)]=\inf_{\blambda, \btheta}\max_{m} S(\Ao_{\bi n}^{\rho^{\bi n}_m}\|\Mo^{\rho^{\bi n}_m}_{\blambda,[\bi n]}),
\end{equation}
where $\rho^{\bi n}_m$ is the eigen-projection of $\bi n\cdot \bi S$ with respect to the eigen-value $m$ and the discretization angles are implicitly contained in $\Mo_{\blambda}$.
When the infimum is realized in a point $\blambda=\blambda^*$, $\btheta=\btheta^*$ we have that $\Mo_{\blambda^*}\Big|_{\btheta=\btheta^*}$ plays the role of optimal approximate joint measurement.
\end{remark}

The upper bound in \eqref{I<} is surely non tight, as it has been obtained by starting from the uniform distribution on the sphere; this can be checked in the explicit cases of small spins given below. However, the role of this bound is at least to say that, when we have a device information loss greater than that, the approximating measurement is not optimal.

By the fact that the device information loss of a covariant approximation is a maximum and has the form \eqref{Delta+S}, we have immediately the following formulation of the MURs for covariant approximate spin measurements.
\begin{remark}[MURs for covariant measurements, second version]\label{rem:genMUR}
The state independent MURs are
\begin{equation}
\forall \Mo\in \Mscr(\Acal_\infty)\quad \forall \Ao_{\bi n}\in\Acal_\infty \qquad \exists \rho \in \Sscr_s  : \ S(\Ao_{\bi n}^\rho\|\Mo^\rho_{[\bi n]})\geq I_s[\Acal_\infty\|\Mscr(\Acal_\infty)]>0;
\end{equation}
such a state $\rho$ is one of the eigen-projections of $\bi n \cdot \bi S$.
\end{remark}

So, in a physical approximate joint measurement $\Mo$ of all the spin components $\Ao_{\bi n}$, $\bi n \in \mathbb S_2$, the loss of information $S\big(\Ao_{\bi n}^\rho\big\|\Mo^\rho_{[\bi n]}\big)$ per direction $\bi n$ can not be arbitrarily reduced. It depends on the state $\rho$ and on the direction $\bi n$, but for every $\bi n$ it can be potentially as large as $I_s[\Acal_\infty\|\Mscr(\Acal_\infty)]$.

We shall compute analytically the minimum information loss in the cases of $s=1/2,\, 1,\, 3/2$. For higher spins, a numerical approach is possible, as the computation has been reduced to the optimization problem \eqref{I+q} over a finite number of real parameters, appearing in integrals \eqref{q+d^2} of known polynomials related to the Wigner small d-matrix (\ref{W-d-m}).

\subsection{Minimum information loss and noisy versions of the target observables}\label{il+nv}

In Section \ref{noise+C} we have seen that the approximating spin components $\Mo_{\blambda, [\bi n]}$ are noisy versions \eqref{Mo-nv} of the target spin components $\Ao_{\bi n}$ with visibility $\eta_{\blambda,\,\btheta}$ \eqref{etavis} and noise $\No^{\blambda,\btheta}_{\bi n}(m)$ \eqref{Nstr}. The visibility \eqref{etavis} was already chosen to be maximal with $\Mo_{\blambda, [\bi n]}$ fixed. Now, we can maximize the visibility also with respect to the class of joint measurements $\Mscr(\Acal_\infty)$ by defining
\begin{equation}\label{mvis}
\eta^*_s:=\sup_{\blambda, \btheta}\eta_{\blambda, \btheta}.
\end{equation}
By comparing this quantity with the result \eqref{I+q} we get $K_s=\eta^*_s$ and
\begin{equation}\label{I+qII}
I_s[\Acal_\infty\|\Mscr(\Acal_\infty)]=\log \frac 1 {\eta^*_s}.
\end{equation}

This equation gives a simple relation between the maximal visibility \eqref{mvis}, \eqref{etavis} and the minimum information loss \eqref{def:Is}, \eqref{I+q} in the case of the spin vector. By our construction,  we have also obtained that, inside the class of covariant measurements $\Mscr(\Acal_\infty)$, to maximize the visibility or to optimize the information loss gives the same optimal measurement. Let us  note that this result is due to the fact that the target observable and the approximating POVM are jointly diagonal.

Our aim in introducing the device information loss \eqref{sie} and the minimum information loss \eqref{def:Is}, \eqref{def:Is2} was to have uncertainty measures, based on information theory, by which MURs could be expressed in a simple way, Sect.\ \ref{sec:MUR}; this construction produced also an incompatibility measure, the minimum information loss.
The result above gives a link with the \emph{robustness measures} \cite{BLPY,HSTZ14,Haap15,HKR15,DFK19,HeiMZ16} which quantify the incompatibility by maximizing the visibility; in other terms, these measures are based on the ability of the target observables to maintain incompatibility against noise.

\subsection{Spin 1/2}\label{sec:ind12}
In this case no free parameter comes out from the angle discretization and the approximate spin components \eqref{kmarg1/2} are very simple.

\begin{theorem}
The device information loss \eqref{sie} and the minimum information loss \eqref{def:Is} turn out to be given by
\begin{equation}\label{Delta_12}
\Delta_{1/2}[\Acal_\infty\|\Mo_{\blambda}]=\log\frac 4{1+2\lambda_{1/2}},
\end{equation}
\begin{equation}\label{I1/2_tutte}
I_{1/2}[\Acal_\infty\|\Mscr(\Acal_\infty)]=S\big(\Zo^{\rho_m}\big\| \Mo_{1/2,[\bi k]}^{\rho_m}\big)=\log \frac 43\simeq 0.415037,
\end{equation}
where $\rho_m=\Zo(m)$.

The first equality in \eqref{I1/2_tutte} shows that $\Mo_{1/2}$ is the optimal measurement in the sense of Remark \ref{optimalMo}; its  marginal in direction $\bi n$ is
\begin{equation}\label{opt1/2}
\Mo_{1/2,[\bi n]}(m)=\frac 12\left[\frac \openone2 +\Ao_{\bi n}(m)\right]=\frac 34\,\Ao_{\bi n}(m)+\frac 14\, \Ao_{\bi n}(-m),
\end{equation}
which is an unbiased noisy version of $\Ao_{\bi n}$ (cf.\ Sect.\ \ref{joint1/2}).
\end{theorem}

\begin{proof}
In this case, by \eqref{q1/2} we have $q(m|\ell,m)=\frac{1+\ell}2$, independent of $m$; then, \eqref{Delta_12} follows from \eqref{Delta+q}.

Directly from the definition \eqref{def:Is} and the expression \eqref{Delta_12} we have
\begin{equation*}
I_{1/2}[\Acal_\infty\|\Mscr(\Acal_\infty)]=\inf_{\lambda_{1/2}\in[0,1]} \Delta_{1/2} [\Acal_\infty\|\Mo_{\blambda}] =  \Delta_{1/2} [\Acal_\infty\|\Mo_{1/2}],
\end{equation*}
and the final expressions in \eqref{I1/2_tutte} follow.

By the facts that there is no freedom in the choice of the $\theta$'s and that the infimum is reached for $\lambda_{1/2}=1$, we get that $\Mo_{ 1/2}$ is the optimal measurement. Then, by \eqref{kmarg1/2} we get the form of the marginal \eqref{opt1/2}.
\end{proof}

In \eqref{opt1/2} we have written the marginal of the optimal measurement in two different ways. Firstly, we have written the noisy version with classical noise, with visibility $1/2$. Then, we have used the expression \eqref{Mo-nv} with general noise and visibility $3/4$; it is this last visibility which is related to our minimum information loss, see \eqref{I+qII}.

Let us remark that, actually, $\Mo_{1/2}$ enjoys a useful additional property. By using the state representation \eqref{state1/2} and the explicit expressions \eqref{prob1/2} for the probabilities, we have
\begin{equation}\label{ent1/2}
S\big(\Ao_{[\bi n]}^\rho\big\|\Mo_{\blambda,[\bi n]}^\rho\big)=s\Big(\lambda_{1/2}- 1/2,\, \bi n\cdot \bi r\Big) ,
\end{equation}
\begin{equation}\label{s(c)}
s(c,x):=\frac{1+x}2\,\log \frac{1+x}{1+c x} + \frac{1-x}2\,\log \frac{1-x}{1-c x}, \qquad \abs c <1, \quad \abs x \leq 1.
\end{equation}
The parameter $\bi r$ is the Bloch vector characterizing the state $\rho$. By taking the $c$-derivative, we see that it is strictly negative,
which implies that $s(c,x)$ decreases when $c$ increases. This means that $\Mo_{ 1/2}$ minimizes \eqref{ent1/2} for any state $\rho$. This peculiarity of the case $s=1/2$ makes possible to state that $\Mo_{1/2}$ is optimal even when we know the system state $\rho$ and to  easily formulate also  a form of state dependent MURs.

\begin{remark}[State dependent MURs]\label{stdep1/2}
The following state dependent bound holds:
\begin{equation}\begin{split}
S\big(\Ao_{\bi n}^\rho\big\| \Mo_{[\bi n]}^\rho\big)\geq S\big(\Ao_{\bi n}^\rho\big\| \Mo_{1/2,[\bi n]}^\rho\big)=\sum_{\epsilon=\pm 1} \frac{1+\epsilon \bi n\cdot \bi r}2\,\log \frac{1+\epsilon \bi n\cdot \bi r}{1+\frac\epsilon 2\,\bi n\cdot \bi r}\,,
\\
\forall \rho\in \Sscr_s, \qquad \forall \Mo\in \Mscr(\Acal_\infty), \qquad \forall \bi n\in\Rbb^3, \quad \abs{\bi n}=1.\end{split}
\end{equation}
\end{remark}

\subsection{Spin 1}\label{sec:ind1}
In this case there is a single parameter \eqref{s1_a} coming from the angle discretization; then, the minimum information loss and the optimal measurement can be computed.
\begin{theorem}\label{prop:s1} Let us set $\rho_m=\Zo(m)$; then,
\begin{equation}\label{I1_tutte}
I_1[\Acal_\infty\|\Mscr(\Acal_\infty)]=  S\big(\Zo^{\rho_m}\big\|\Mo^{\rho_m}_{1,[\bi k]}\big)\Big|_{ a=a_0} =\log \frac 2{a_0\left(3-a_0^2\right)}\simeq 0.682505.
\end{equation}
The quantity $a_0$ is the real solution of the equation
\begin{equation}\label{a-eq}
a^3-a^2-5a+\frac 73=0,
\end{equation}
which is given by
\begin{equation}\label{a-properties}
a_0=\frac 13 \left(1+8\cos \alpha\right), \qquad \cos(3\alpha -\pi)=\frac 18, \qquad \alpha\in(0,\pi/2).
\end{equation}
This gives also
\begin{equation}\label{a-properties2}
a_0\simeq 0.444703; \qquad \cos 3\alpha=-\frac 18,  \qquad (\cos\alpha)^3=\frac 14 \left(3\cos\alpha -\frac 18\right).
\end{equation}

The optimal measurement is $\Mo_1\big|_{a=a_0}$ and its  marginal along
$\bi k$  is given by
\begin{equation}\label{opt_1}\begin{split}
&\Mo_{ 1,[\bi k]}(m)\big|_{a=a_0}=\eta^*_1\Zo(m) +(1-\eta^*_1) \No^*_{\bi k}(m),
\\
&\No^*_{\bi k}(0)=(1-\kappa)\left[\Zo(1)+\Zo(-1)\right], \qquad \No^*_{\bi k}(\pm1)=\frac 12\, \Zo(0)+\kappa \Zo(\mp 1)
\\
&\eta^*_1=\frac {a_0} 2 \left(3-a^2_0\right)
\simeq 0.623083, \qquad 0<\kappa = \frac{(1-a_0)^3}{4(1-3a_0+a_0^3)}
<1.\end{split}
\end{equation}

\end{theorem}

\begin{proof}
From \eqref{q-1} we have
\begin{equation*}
\sum_{\ell}\lambda_{\ell}q_a(\pm 1|\ell,\pm 1)=\lambda_+\left[1-\frac{(1+a)^3}8\right]+\lambda_-\,\frac{(1-a)^3}8+\lambda_0\,\frac{(2+a) (1-a)^2}4,
\end{equation*}
\[
\sum_{\ell}\lambda_{\ell}q_a(0|\ell,0)=\left(\lambda_++\lambda_-\right)\frac a 2 \left(3-a^2\right)+\lambda_0a^3.
\]
One can check that both these expressions have an absolute maximum in $\lambda_+=1$ for all $a\in (0,1)$. Then, \eqref{I+q} gives
\begin{equation*}
K_1\leq \sup_{a\in(0,1)}\min_m\sup_{\blambda}\sum_{\ell}\lambda_{\ell}q_a(m|\ell,m)= \sup_{a\in(0,1)}\min_mq_a(m|1,m).
\end{equation*}
On the other side, by eliminating the supremum over the $\lambda$'s and choosing $\lambda_{\ell}=\delta_{\ell,1}$ in \eqref{I+q}, we get
$
K_1\geq  \sup_{a\in(0,1)}\min_mq_a(m|1,m)$; so, the equality holds and we have
\[
K_1=\sup_a\min_mq_a(m|1,m)=\sup_a\min \left\{ 1-\frac{(1+a)^3}8,\; \frac a 2 \left(3-a^2\right)\right\}.
\]

The first term in the minimum decreases with $a$ and the second one increases; this means that the supremum over $a$ is reached when these two terms are equal, which happens when \eqref{a-eq} holds. This proves \eqref{I1_tutte}. It is possible to check that \eqref{a-properties} is the unique real solution of \eqref{a-eq} and that this gives the properties \eqref{a-properties2}.

Equation \eqref{I1_tutte} implies also that the optimal measurement is $\Mo_1\big|_{a=a_0}$. By inserting  $a_0$ into the expression \eqref{M11} of its $\bi k$-marginal we get the expressions \eqref{opt_1}.
\end{proof}

\begin{remark}
Differently from the case $s=1/2$, for $s=1$ the marginal  $\Mo_{1,[\bi k]}\big|_{a=a_0}$ of the optimal measurement  is not unbiased because $a_0\neq 1/3$. Indeed, on the maximally mixed state $\rho_0$, the relative entropy is not zero and its value is
\[
S\big(\Ao_{[\bi n]}^{\rho_0}\big\|\Mo_{1,[\bi n]}^{\rho_0}\big)_{a=a_0}=\frac 23\log \frac 6{4+3a_0(1-a_0)} +\frac 13\log \frac 6{3a_0(10-a_0)-5}\simeq 0.103607.
\]
\end{remark}

\subsection{Spin 3/2}\label{sec:ind32}
\begin{theorem} \label{theor:3/2} Let us set $\rho_m=\Zo(m)$; then, we have
\begin{equation}\label{Iall3/2}
I_{3/2}[\Acal_\infty\|\Mscr(\Acal_\infty)]=S\big( \Zo^{\rho_m}\big\|\Mo_{3/2,[\bi k]}^{\rho_m}\big)_{a=a_0}=\log \frac {32}{ 45 -24 a_0-24a_0^2 -8a_0^3}\simeq 0.88615563;
\end{equation}
$\Mo_{3/2}\big|_{a=a_0}$ is the optimal measurement. The quantity $a_0$ is the unique real solution in $(0,1)$ of the equation
\begin{equation}\label{a-eq3/2}
a^4 -6 a^2-8 a+\frac{15}2=0,
\end{equation}
which gives
\begin{equation}\label{a03/2value}
a_0\simeq 0.6461537831.
\end{equation}

\end{theorem}

\begin{proof}
From \ref{app:spin3/2} we get
\[
\max_{\ell}q_a(\pm3/2|\ell,\pm3/2)=q_a(\pm3/2|3/2,\pm3/2)=\frac 1{16}\left(15-4a- 6a^2-4a^3- a^4\right),
\]
a quantity which decreases with $a$ from $\frac {15}{16}$ to 0, and
\[
\max_{\ell}q_a(\pm 1/2|\ell,\pm 1/2)=q_a(\pm 1/2|3/2,\pm 1/2)=\frac 1{16} \left( 12a+ 6a^2-4a^3-3 a^4\right),
\]
a quantity which increases with $a$ from $0$ to $\frac {11}{16}$. Then, as in the proof of Theorem \ref{prop:s1}, we get
\begin{multline*}
K_{3/2}=\sup_{a\in(0,1)} \min_m q_a(m|3/2,m)\\ {}=\sup_{a\in(0,1)}\frac 1{16}\min\left\{15-4a- 6a^2-4a^3- a^4,\,12a+ 6a^2-4a^3-3 a^4\right\}.
\end{multline*}
By equating these two expressions we get equation \eqref{a-eq3/2}, whose solution \eqref{a03/2value} is computed numerically.
As we have
\[
\min_mq_a(m|3/2,m)=\begin{cases}q_a(\pm 1/2|3/2,\pm 1/2) & \text{for} \ a\leq a_0,
\\ q_a(\pm 3/2|3/2,\pm 3/2) & \text{for}\  a\geq a_0,\end{cases}
\]
\eqref{I+q} gives
\begin{equation*}
I_{3/2}[\Acal_\infty\|\Mscr(\Acal_\infty)]=\log \left(q_{a_0}(m|3/2,m)\right)^{-1};
\end{equation*}
by using also \eqref{a-eq3/2}, the final expression in \eqref{Iall3/2} follows.
By Theorem \eqref{Iprop}, the optimal measurement is identified and the intermediate expression in \eqref{Iall3/2} follows.
\end{proof}

By comparing \eqref{I+qII} and \eqref{Iall3/2}, we have that the optimal visibility is
\[
\eta^*_{3/2}=\frac14\left(\frac{45}8 -3a_0-3a_0^2-a_0^3\right)\simeq 0.541054
\]
with $a_0$ given in Theorem \ref{theor:3/2}. Also the expression of the optimal noise could be obtained, but it would be involved and we do not give explicitly here.

By direct computations one can check that
the optimal measurement is biased and that on the maximally mixed state $\rho_0$ it gives
\begin{equation}\label{bias3/2}
S\big( \Zo^{\rho_0}\big\|\Mo_{3/2,[\bi k]}^{\rho_0}\big)_{a=a_0}=\frac 12 \,\log [4a_0(1-a_0)]^{-1}\simeq 0.0644281.
\end{equation}

\begin{remark}\label{order1}
The results  we have found for small spin values give
\begin{equation}\label{<<<}
0<I_{1/2}[\Acal_\infty\|\Mscr(\Acal_\infty)]<I_{1}[\Acal_\infty\|\Mscr(\Acal_\infty)]<I_{3/2}[\Acal_\infty\|\Mscr(\Acal_\infty)].
\end{equation}
This chain of inequalities suggests the conjecture that $I_{s}[\Acal_\infty\|\Mscr(\Acal_\infty)]$ could grow with $s$: in some sense the minimum  information loss grows with the complexity of the spin system.
\end{remark}

\section{MURs for two and three orthogonal components}\label{sec:2+3ort}
In this section we study the MURs for the cases of two and three orthogonal spin components. As remarked in \cite{DamSW15}, it is not possible to get the case of infinite components from the case of three orthogonal components; only the case of infinite components respects the rotation symmetry, while in the other case the three directions are fixed.  The cases of orthogonal components involve less symmetries and there is more freedom in the construction of the approximate joint measurements; so it is meaningful to enlighten the differences between the case of the spin components in all directions and the case of orthogonal components. In principle also a few non-orthogonal components could be considered; in \cite{BGT18} we already considered two non-orthogonal spin components with $s=1/2$, but with the sum of relative entropies as starting point.

The cases of orthogonal components allow to show how the minimum information loss and the related MURs can be introduced also for other sets of observables by adapting the construction of Section \ref{sec:newind}.
Moreover, the minimum information loss can be used as quantification of the incompatibility of the target observables and allows to compare different sets of observables. In the cases of spin components we shall obtain orderings for different numbers of target observables and different values of $s$, which are not at all trivial or intuitive.

\subsection{Target observables and approximate joint measurements}\label{sec:2+3symm}
The first set of target observables we consider is $\Acal_3=\{\Xo,\Yo,\Zo\}$, which is covariant with respect to the octahedron group $O$, see \ref{app:3orth}. Then, $\Mscr(\Acal_3)$ is the set of observables with value space $\Xscr^3$ and $O$-covariant in the sense of \eqref{covMo3}. By using the notation \eqref{gen_marg} and the covariance properties \eqref{--}, \eqref{symm_marg}, \eqref{covMo3} we have that
\begin{equation}\label{infty_sub3}
\Mo\in\Mscr(\Acal_\infty) \quad \Rightarrow \quad \Mo_{[\bi i,\bi j,\bi k]}\in \Mscr(\Acal_3).
\end{equation}

The other set of target observables is $\Acal_2=\{\Xo,\Yo\}$, which is covariant with respect to the  dihedral group $D_4$, see \ref{app:2orth}. Then, $\Mscr(\Acal_2)$ is the set of observables with value space $\Xscr^2$ and $D_4$-covariant in the sense of \eqref{covMo2}. By using the notation \eqref{gen_marg} and the covariance properties \eqref{covMo2}, \eqref{covMo3} we have that
\begin{equation}\label{3_sub2}
\Mo\in\Mscr(\Acal_3) \quad \Rightarrow \quad \Mo_{[\bi i,\bi j]}\in \Mscr(\Acal_2).
\end{equation}

Note that the implications above are one-sided: there are elements in $\Mscr(\Acal_2)$ which are not marginals of elements in $\Mscr(\Acal_3)$ and the same for $\Mscr(\Acal_3)$ with respect to $\Mscr(\Acal_\infty)$.

We obtained the explicit form of a covariant approximate joint measurement, for two and three orthogonal components, only in the case of a spin 1/2. For a generic spin $s$ we can give only particular covariant approximate joint measurements, such as the ones based on optimal cloning.

\subsubsection{Optimal cloning and approximate joint measurements.}\label{sect:cl}
As approximate joint measurement of the spin components $\Ao_h$, $h=1,\ldots,\mathtt{r}$, a significant multi-observable $\Mo_{\rm cl}\in \Mscr(\Xscr^\mathtt{r})$ can be constructed by using the so called \emph{optimal cloning} \cite{Wer98,KW99,HSTZ14}; its univariate marginals are given by \eqref{eq:maropt}. Let us stress that the marginal of the multi-observable constructed by optimal cloning can be seen as a noisy version of the target observable, with classical noise; however, this decomposition is not unique, as in the case of infinite components.

When the target observables are $\Acal_3=\{\Xo,\Yo,\Zo\}$, we get the multi-observable $\Mo_{\rm cl}^3$, whose univariate marginals \eqref{eq:maropt} take the form
\begin{equation}\label{eq:3cl}
\Mo_{{\rm cl}[i]}^3(m)=\frac 1{3(s+1)}\left[ \openone +\left(s+2\right)\Xo_i(m)\right], \qquad i=1,2,3, \quad m\in\Xscr.
\end{equation}
Obviously $\Mo_{\rm cl}^3\in \Mscr(\Xscr^3)$, but one has also $\Mo_{\rm cl}^3\in \Mscr(\Acal_3)$, as shown in  \ref{app:clon}.

When the target observables are $\Acal_2=\{\Xo,\Yo\}$, the optimal cloning gives the bi-observable $\Mo_{\rm cl}^2\in \Mscr(\Xscr^2)$ and   \eqref{eq:maropt} becomes
\begin{equation}\label{cl2:marg}
\Mo_{{\rm cl}[i]}^2(m)=\frac 1{4(s+1)}\left[ \openone +\left(2s+3\right)\Xo_i(m)\right], \qquad i=1,2,\quad m\in\Xscr.
\end{equation}
Again one has also $\Mo_{\rm cl}^2\in \Mscr(\Acal_2)$, as shown in  \ref{app:clon}.

\subsubsection{Spin 1/2.}\label{1/2joint_ort}

For a spin $1/2$ the explicit expressions of the general element in $\Mscr(\Acal_3)$ and $\Mscr(\Acal_2)$ have been obtained in \cite[Proposition 5, Theorem 10]{BGT18} and used also in \cite{BarG18}.
Then, the most general covariant joint measurement in $\Mscr(\Acal_3)$ \cite[Eq.\ (11)]{BarG18}  can be written as
\begin{equation}\label{3cl12}
\Mo_c(m_1,m_2,m_3)=\frac\openone 8 + \frac c 2 \left(m_1 S_x+m_2 S_y+m_3 S_z\right), \quad \abs c \leq\frac 1 {\sqrt 3}.
\end{equation}
Similarly, the most general element in  $\Mscr(\Acal_2)$  has the expression
\cite[Eq.\ (7)]{BarG18}
\begin{equation}\label{2cl12}
\Mo_c(m_1,m_2)=\frac\openone 4 + c \left(m_1 S_x+m_2 S_y\right), \qquad \abs c \leq\frac 1 {\sqrt 2} .
\end{equation}

\begin{remark} In both the cases of two and three orthogonal components, the univariate marginals have the expression
\begin{equation}\label{marg1/2}
\Mo_{c[i]}(m)=\frac \openone2 +2cmS_i=\begin{cases} c\Xo_i(m)+\left(1-c\right)\frac \openone 2,  &c\geq 0, \\
\abs c \Xo_i(-m)+\left(1-\abs c\right)\frac \openone 2,  &c<0;
\end{cases}
\end{equation}
the only difference is the maximally possible value for $\abs c$: $\abs c \leq 1/\sqrt 3$ in the case of three components and $\abs c \leq 1/\sqrt 2$ in the case of two components. Also
the marginal of the optimal measurement \eqref{opt1/2} for infinite components has the form \eqref{marg1/2} with $c=1/2$.
\end{remark}

\begin{remark}  By particularizing \eqref{eq:3cl} and \eqref{cl2:marg} to $s=1/2$, we obtain that the marginals of the joint measurements from optimal cloning have again the form \eqref{marg1/2} with $c=5/9$ in the case of three components and $c=2/3$ in the case of two components.
As we have $1/2<5/9<1/\sqrt 3<2/3<1/\sqrt 2$, there is an increase of minimum classical noise in going from the case of two orthogonal components, to cloning of two components, three components, cloning of three components, infinite components.
\end{remark}

\subsection{The information loss}

Analogously to what is done in Section \ref{sec:newind}, also in the case of orthogonal spin components it is possible to define the device information loss  and the minimum information loss.
The \emph{device information loss} of $\Mo$ is defined as in \eqref{sie}; then, exactly as for \eqref{ntok}, after the supremum on the states, the covariance implies the independence from the direction. So, we have: for $\mathtt{r}=2,\,3$,
\begin{equation}\label{sie2+3}
\Delta_s[\Acal_\mathtt{r}\|\Mo]:=\sup_{\rho\in\Sscr_s, \; i: i\leq \mathtt{r}}S\left(\Xo_{i}^\rho\|\Mo_{[i]}^\rho\right)=\sup_{\rho\in\Sscr_s}S\left(\Xo_{i}^\rho\|\Mo_{[i]}^\rho\right), \qquad \Mo\in \Mscr(\Acal_\mathtt{r}).
\end{equation}

By optimizing over the approximate joint measurement $\Mo$ we get the \emph{minimum information loss}
\begin{equation}\label{def:Is2+3}
I_s[\Acal_\mathtt{r}\|\Mscr(\Acal_\mathtt{r})]:=\inf_{\Mo\in\Mscr(\Acal_\mathtt{r})}\Delta_s[\Acal_\mathtt{r}\|\Mo]=\inf_{\Mo\in\Mscr(\Acal_\mathtt{r})} \sup_{\rho\in\Sscr_s} S\left(\Xo_{i}^\rho\|\Mo_{[i]}^\rho\right), \qquad \mathtt{r}=2,3.
\end{equation}

As done in Section \ref{sec:mil} and in \cite{RS19,BGT18}, we can extend the previous definitions to non-symmetric approximate joint measurements, without changing the final conclusions. Firstly, we introduce the device information loss for general measurements:
\begin{equation}\label{def:DeltaGEN}
\Delta_s[\Acal_\mathtt{r}\|\Mo]=\sup_{\rho\in\Sscr_s, \; i: i\leq \mathtt{r}}S\left(\Xo_{i}^\rho\|\Mo_{[i]}^\rho\right), \qquad \Mo\in \Mscr(\Xscr^\mathtt{r}), \qquad \mathtt{r}=2,3.
\end{equation}
Obviously, now we cannot eliminate the maximum over the directions as in \eqref{sie2+3}, because this follows from the covariance. Then, we optimize over all these measurements by defining
\begin{equation}\label{def:IsGEN}
I_s[\Acal_\mathtt{r}\|\Mscr(\Xscr^\mathtt{r})]:=\inf_{\Mo\in\Mscr(\Xscr^\mathtt{r})}\Delta_s[\Acal_\mathtt{r}\|\Mo]=\inf_{\Mo\in\Mscr(\Xscr^\mathtt{r})} \sup_{\rho\in\Sscr_s, \; i: i\leq \mathtt{r}} S\left(\Xo_{i}^\rho\|\Mo_{[i]}^\rho\right), \qquad \mathtt{r}=2,3.
\end{equation}

Next proposition shows that this extension does not change the value of the minimum information loss and that this value grows with the increasing complexity of the set of observables, i.e.\ going from $\Acal_2$, to $\Acal_3$, and then to $\Acal_\infty$

\begin{proposition}
\label{prop:2+3}
The two definitions \eqref{def:Is2+3} and \eqref{def:IsGEN} are equivalent, as we have
\begin{equation}\label{I=Ir}
I_s[\Acal_\mathtt{r}\|\Mscr(\Xscr^\mathtt{r})]=I_s[\Acal_\mathtt{r}\|\Mscr(\Acal_\mathtt{r})], \qquad \mathtt{r}=2,3.
\end{equation}

Moreover, the minimum information loss is strictly positive and finite and we have
\begin{equation}\label{A23infty}
0<I_s[\Acal_2\|\Mscr(\Acal_2)] \leq I_s[\Acal_3\|\Mscr(\Acal_3)]\leq I_s[\Acal_\infty\|\Mscr(\Acal_\infty)]<+\infty.
\end{equation}
\end{proposition}

\begin{proof} The proof of \eqref{I=Ir} is a very slight modification of what is done in \cite{BGT18}. Let us use the notation $G_3=O$ and $G_2=D_4$ for the two groups introduced in \ref{app:3orth} and \ref{app:2orth}; the actions of these two groups on the POVMs, as given in the two appendices, can be seen to satisfy the hypotheses of Theorem 9 of \cite{BGT18}, as done in \cite[Sections B.2, B.4]{BGT18}. We denote by $g\Mo$ the action of an element $g\in G_\mathtt{r}$ on the POVM $\Mo\in \Mscr(\Xscr^\mathtt{r})$ and by $\Mo_{G_\mathtt{r}}\in \Mscr(\Acal_\mathtt{r})$ the \emph{covariant version} of $\Mo$ as done in \cite[Sections 3.1, 4.1]{BGT18}. Thanks to the hypotheses on the group action of \cite[Theorem 9]{BGT18}, by substituting the sum of the relative entropies by their maximum, we get that the results on the \emph{entropic divergence} of Theorems 4 and 9 of \cite{BGT18} go into analogous results on the device information loss. In this way one proves that, for $\mathtt{r}=2,3$,
\[
\Delta_s[\Acal_\mathtt{r}\|g\Mo]=\Delta_s[\Acal_\mathtt{r}\|\Mo], \qquad \forall g\in G_\mathtt{r}, \qquad \forall \Mo\in \Mscr(\Xscr^\mathtt{r}),
\]
\[
\Delta_s\big[\Acal_\mathtt{r}\big\|\Mo_{G_\mathtt{r}}\big]\leq \Delta_s[\Acal_\mathtt{r}\|\Mo], \qquad \forall \Mo\in \Mscr(\Xscr^\mathtt{r}).
\]
As $\Mo_{G_\mathtt{r}}\in \Mscr(\Acal_\mathtt{r})$, by taking the infimum we get \eqref{I=Ir}.

To prove \eqref{A23infty}, note that,  by \eqref{infty_sub3} and \eqref{3_sub2}, the definition \eqref{def:Is2+3} gives the ordering among the three information losses $I_s[\Acal_\mathtt{r}\|\Mscr(\Acal_\mathtt{r})]$, $\mathtt{r}=2,3,\infty$. We already proved the last inequality in Theorem \ref{Iprop}, cf.\ the upper bound in \eqref{I<}. The proof of the strict positivity is analogous to the proof of the strict positivity in \eqref{I<}. Exactly as in the final part of the proof of Theorem \ref{Iprop} we obtain $0<c_{\rm inc}(\Xo,\Yo)\leq 2 I_s[\Acal_2\|\Mscr(\Acal_2)]$, where $c_{\rm inc}(\Xo,\Yo)$ is  defined in \cite[(10)]{BGT18}.
\end{proof}

\subsubsection{Entropic MURs.}\label{MUR2+3} By the definition and the strict positivity of the minimum information loss we get the state independent MURs in a formulation involving the device information loss:
\begin{equation}\label{rform1}
\Delta_s[\Acal_\mathtt{r}\|\Mo] \geq I_s[\Acal_\mathtt{r}\|\Mscr(\Acal_\mathtt{r})]>0, \qquad \forall \Mo\in \Mscr(\Xscr^\mathtt{r})\supset\Mscr(\Acal_\mathtt{r}).
\end{equation}
We have used \eqref{I=Ir} to extend the set of possible measurements $\Mo$. This form of MURs is the analogue of what is done in
Remark \ref{MURsfirst} for the case of infinitely many components.

By proving that the supremum over the states in \eqref{sie2+3} reduces to a maximum, we could get a MUR formulation analogous to  the one in Remark \ref{rem:genMUR}, but we skip this.

\subsubsection{Spin 1/2.}\label{sec:orth_s1/2} By using the state representation \eqref{state1/2} and the univariate measure \eqref{marg1/2}, we can compute the relative entropies, as done in equations \eqref{ent1/2} and \eqref{s(c)}. Then, by taking the supremum over the states, we get
\begin{equation} \label{Delta/k}
\Delta_{1/2}[\Acal_\mathtt{r}\|\Mo_c]=S\big(\Xo_i^{\rho_i}\big\|(\Mo_c)^{\rho_i}_{[i]}\big)=\log\frac 2{1+c}, \qquad \abs c \leq \frac 1{\sqrt \mathtt{r}}, \quad \mathtt{r}=2,3.
\end{equation}
Here, the measurement $\Mo_c$ is given by \eqref{3cl12} for $\mathtt{r}=3$ or by \eqref{2cl12} for $\mathtt{r}=2$, while the state $\rho_i$ is anyone of the two eigen-projections of $S_i$.

By the definition \eqref{def:Is2+3} and the explicit expression \eqref{Delta/k}, we obtain
\begin{equation}\label{I1/2_r}
I_{1/2}[\Acal_\mathtt{r}\|\Mscr(\Acal_\mathtt{r})]=\inf_{c\in [-1/\sqrt \mathtt{r}, 1/\sqrt\mathtt{r}]}S\big(\Xo_i^{\rho_i}\big\|(\Mo_c)^{\rho_i}_{[i]}\big) =\log \frac 2{1+1/\sqrt\mathtt{r}}\,, \qquad \mathtt{r}=2,3.
\end{equation}

Let us note that there is an optimal POVM, the one with $c=1/\sqrt \mathtt{r}$, the same of the one appearing in \cite{BGT18,Jost+19,HeiMZ16}, where different optimality criteria where used. By using this measurement it would be possible to give a state dependent version of the MURs as done in Remark \ref{stdep1/2}.

\subsubsection{The bounds from optimal cloning.}
For $s>1/2$ we can get a bound on the minimal information loss by using the POVM obtained from optimal cloning, because by construction we have
\begin{equation}\label{bound:cl}
I_{s}[\Acal_\mathtt{r}\|\Mscr(\Acal_\mathtt{r})]\leq \Delta_s[\Acal_\mathtt{r}\|\Mo_{\rm cl}^\mathtt{r}], \qquad \mathtt{r}=2,3.
\end{equation}

\paragraph{Three orthogonal components.} Let us set $p_m:=\Xo^\rho(m)$; then, by \eqref{eq:3cl} and \eqref{relentMon}, we get
\[ 
\Mo_{\rm cl[1]}^{3,\,\rho}(m)=\frac {1+\left(s+2\right)p_m}{3(s+1)},
\qquad S\big(\Xo^\rho\big\|\Mo_{\rm cl[1]}^{3,\,\rho}\big)=\sum_{m=-s}^s p_m\log \frac{3\left(s+1\right)p_m}{1+\left(s+2\right)p_m}.
\]
This gives the device information loss
\begin{equation}\label{Delta_3cl}
\Delta_s[\Acal_3\|\Mo_{\rm cl}^3]=\sup_\rho S\big(\Xo^\rho\big\|\Mo_{\rm cl[1]}^{3,\,\rho}\big)=\log \frac{3(s+1)}{s+3}.
\end{equation}

\paragraph{Two orthogonal spin components.} By the same definition of $p_m$ and using
\eqref{cl2:marg} instead of \eqref{eq:3cl}, in a similar way we get
\[ 
\Mo_{\rm cl[1]}^{2,\,\rho}(m)=\frac {1+\left(2s+3\right)p_m}{4(s+1)}, \qquad
S\big(\Xo^\rho\big\|\Mo_{\rm cl[1]}^{2,\,\rho}\big)=\sum_{m=-s}^s p_m\log \frac{4\left(s+1\right)p_m}{1+\left(2s+3\right)p_m},
\]
\begin{equation}\label{Delta_2cl}
\Delta_s[\Acal_2\|\Mo_{\rm cl}^2]=\sup_\rho S\big(\Xo^\rho\big\|\Mo_{\rm cl[1]}^{2,\,\rho}\big)=\log \frac{2(s+1)}{s+2}.
\end{equation}

Note that the device information losses \eqref{Delta_3cl} and \eqref{Delta_2cl} grow with $s$ and that they enjoy some unexpected relations, such as
\[
\Delta_1[\Acal_2\|\Mo_{\rm cl}^2]> \Delta_{1/2}[\Acal_3\|\Mo_{\rm cl}^3],
\qquad
\Delta_2[\Acal_2\|\Mo_{\rm cl}^2]= \Delta_{1}[\Acal_3\|\Mo_{\rm cl}^3],
\]
\[
\lim_{s\to +\infty}\Delta_s[\Acal_2\|\Mo_{\rm cl}^2]= \Delta_{3}[\Acal_3\|\Mo_{\rm cl}^3].
\]
For instance, the first relation says that, for the devices constructed by optimal cloning, the information loss for the case of two orthogonal components and $s=1$ is greater than the information loss for the case of three orthogonal components and $s=1/2$.

\subsection{Some orderings and bounds}\label{sec:order}
As we already said, the minimum information loss can be interpreted as a quantification of the incompatibility of the set of target observables. So, we can take the results obtained on $I_{s}[\Acal_\mathtt{r}\|\Mscr(\Acal_\mathtt{r})]$, \ $\mathtt{r}=2,3, \infty$, \ $s=1/2, 1,3/2,\ldots$,
to compare different sets of spin observables (even in different Hilbert spaces) from the point of view of incompatibility; as we shall see, some non intuitive relations appear.

First of all we have the inequalities  \eqref{<<<} in the case of all the components and small $s$; for the same $s$ and different $\mathtt{r}$ we have the inequalities \eqref{A23infty}.

By  the optimal cloning bound \eqref{bound:cl} and the growing with $s$ of the expressions \eqref{Delta_3cl} and \eqref{Delta_2cl}, we get the bounds
\begin{equation}\label{Ibounds}\begin{split}
I_{s}[\Acal_2\|\Mscr(\Acal_2)]\leq 1, &\qquad I_{s}[\Acal_3\|\Mscr(\Acal_3)]\leq \log 3, \qquad s\geq \frac 12\,,
\\
&I_{s}[\Acal_3\|\Mscr(\Acal_3)]\leq 1, \qquad \frac 12\leq s\leq 3.\end{split}
\end{equation}

By the bound \eqref{bound:cl} again, and the fact the we have the numerical value of  $I_{s}[\Acal_\infty\|\Mscr(\Acal_\infty)]$ for $s=1,\,3/2$, see equations \eqref{I1_tutte} and \eqref{Iall3/2}, we obtain
\begin{equation}\label{Iorder}\begin{split}
I_{s}[\Acal_2\|\Mscr(\Acal_2)]\leq I_{1}[\Acal_\infty\|\Mscr(\Acal_\infty)] <1, \qquad &1/2\leq s\leq 3,
\\
I_{s}[\Acal_3\|\Mscr(\Acal_3)]< I_{1}[\Acal_\infty\|\Mscr(\Acal_\infty)],\qquad  &s=1/2,\,1,
\\
I_{s}[\Acal_2\|\Mscr(\Acal_2)]< I_{3/2}[\Acal_\infty\|\Mscr(\Acal_\infty)],\qquad  &1/2\leq s\leq 11,
\\
I_{s}[\Acal_3\|\Mscr(\Acal_3)]< I_{3/2}[\Acal_\infty\|\Mscr(\Acal_\infty)],\qquad  &1/2\leq s\leq 2.\end{split}
\end{equation}
For instance, the second-last inequality says that two orthogonal components for $s=11$ are less incompatible than the set of all components for $s=3/2$; similar interpretations hold for the other inequalities.

\subsection{Noise and visibility} \label{N+v}

The marginals of the optimal measurements for spin $1/2$, \eqref{marg1/2} with $c=1/\sqrt\mathtt{r}$, $\mathtt{r}=2,3$, can be written in a way similar to \eqref{Mo-nv}: for $\mathtt{r}=2,3$,
\[ 
\Mo_{c[i]}(m)\big|_{c=1/\sqrt\mathtt{r}}= \eta^{\mathtt{r}}_{1/2}\Xo_i(m)+\left(1-\eta^{\mathtt{r}}_{1/2}\right)\Xo_i(-m),\qquad \eta^{\mathtt{r}}_{1/2}=\frac 12\left(1+\frac1{\sqrt\mathtt{r}}\right);
\]
the same holds for the marginals
\eqref{eq:3cl}, \eqref{cl2:marg} of the joint measurements generated by optimal cloning:
\[
\Mo_{{\rm cl}[i]}^{\mathtt{r}}(m)=\eta^{\mathtt{r}}_{{\rm cl},\, s}\Xo_{i}(m)+\left(1-\eta^{\mathtt{r}}_{{\rm cl},\, s}\right)\frac{\openone- \Xo_{i}(m)}{2s}, \qquad \eta^{\mathtt{r}}_{{\rm cl},\, s}=\frac{s+\mathtt{r}}{\mathtt{r}(s+1)}.
\]
Then, from  \eqref{I1/2_r}, \eqref{Delta_3cl}, \eqref{Delta_2cl}, we get
\begin{equation} \label{nr=il}
\Delta_s[\Acal_{\mathtt{r}}\|\Mo_{\rm cl}^{\mathtt{r}}]=\log \left(\eta^{\mathtt{r}}_{{\rm cl},\, s}\right)^{-1}, \qquad I_{1/2}[\Acal_\mathtt{r}\|\Mscr(\Acal_\mathtt{r})] =\log \left(\eta^{\mathtt{r}}_{1/2}\right)^{-1}, \quad \mathtt{r}=2,3.
\end{equation}

The visibilities above have been obtained by allowing for general noises, not only classical ones.
Inside the \emph{noise robustness} approach to incompatibility, the two visibilities $\eta^{\mathtt{r}}_{1/2}$ for spin 1/2 have already been obtained in \cite{DFK19}; they are in the class called \emph{incompatibility generalized robustness}, which means that general POVMs are allowed for noises.
By comparing with Section \ref{il+nv}, we can say that we have shown how to generalize this approach to the case of infinitely many observables, such as the spin vector. Moreover, by using information loss measures, we have shown how to link this problem with the one of uncertainty measures and MURs. Let us also stress that formulae like \eqref{nr=il} and \eqref{I+qII} hold in this particular cases; they have not a general validity. The case of non-orthogonal spin components \cite{BGT18,DFK19} could be a promising test to see the differences. In principle, our minimum information loss does not relay on the noisy versions of the target observables.

\section{Conclusions} \label{sec:concl}

The entropic formulation of MURs has the advantage of being well based on information theory (in particular on the notion of information loss) and independent of the measurement units of the observed physical quantities and from a reordering of their possible values \cite{BGT18,BGT17,BarG18}.
By using the case of the spin components, in this article we have shown that the approach based on the relative entropy can be extended so to treat on the same footing finitely or infinitely many observables and that a quantitative uncertainty bound can be constructed.

By introducing the worst information loss with respect to the target observables and the system states, we have defined the \emph{device information loss} in the various cases \eqref{sie}, \eqref{sie2+3}, \eqref{def:DeltaGEN}.
Then, by optimizing with respect to the approximating joint measurements we have defined the \emph{minimum information loss} \eqref{def:Is}, \eqref{def:Is2+3}, \eqref{def:IsGEN}. These two quantities allow for a clear formulation of state independent MURs, see Sections \ref{sec:MUR} and \ref{MUR2+3}.

To realize the minimum information loss one needs also to optimize the approximating measurement; an interesting point is that the ``best'' approximating measurement of a target spin observable is not necessarily a noisy version of the target, with classical noise, but most general noise structures can be involved, as discussed in Sections \ref{noise+C}, \ref{il+nv}, \ref{N+v}.

Moreover, the lower bound appearing in the state independent MURs, the minimum information loss, plays also the role of measure of incompatibility and allows to order different sets of target observables according to increasing incompatibility, as done in the inequalities \eqref{<<<}, \eqref{A23infty}, \eqref{Iorder}.

However, the computations of the two ``information losses'' need to solve difficult optimization problems and we have done these computations only for small values of $s$, Sections \ref{sec:ind12}, \ref{sec:ind1}, \ref{sec:ind32}, \ref{sec:orth_s1/2}. To compute the minimum information loss for other values of the spin also numerical computations should be surely involved.

Another open problem is the conjecture given after inequality \eqref{<<<}: is it true that the minimum information loss grows with $s$? For the cases of two and three orthogonal components we proved that the minimum information loss is upper  bounded by a value independent from $s$, see \eqref{Ibounds}. However, for the case of infinitely many components we proved only the existence of the upper bound \eqref{I<}, which grows with $s$; the problem of the asymptotic behaviour of $I_s[\Acal_\infty\|\Mscr(\Acal_\infty)]$ for large $s$ is open.

As we remarked at the end of Section \ref{sec:MUR}, the proof of MURs is independent of the choice of the class of approximating joint measurements. Anyway, the value of the minimum information loss can depend on this choice. Another open problem is to study if the lower bound remains $I_s[\Acal_\infty\|\Mscr(\Acal_\infty)] $ even with classes of measurements larger than $\Mscr_\infty$.
Indeed, one could consider post-processing procedures different from our, or even general POVMs on $\Xscr^{\mathbb S_2}$ that are not even constructed by post-processing of a POVM on $\mathbb S_2$. Our conjecture is that even these more general POVMs cannot give a lower information loss.

\appendix

\section{Spin $s$: rotations and $q$-coefficients}\label{app:symm}

Let us consider the rotation group in $\Rbb^3$: a counterclockwise rotation of the angle $\alpha$ around the unit vector $\bi u$ is denoted by
\begin{equation}\label{R_u}
R_{\bi u}(\alpha)\in SO(3), \qquad \abs{\bi u}=1, \quad \alpha\in [0,2\pi).
\end{equation}
Then, we introduce the unitary representation of $SO(3)$ on $\Hscr=\Cbb^{2s+1}$, given by
\begin{equation}\label{U(R)}
U\big(R_{\bi u}(\alpha) \big):=\exp\left\{-\rmi \alpha\, \bi u \cdot \bi S\right\}.
\end{equation}
Such a representation is an essential tool in our whole construction; this representation and its main properties can be found, e.g.,  in \cite[Sect.\ 3.5]{BieL81}, \cite[Sect.\ 3.11]{Hol11}.

By comparing equations \eqref{U(R)} and \eqref{Sphi}, we have the identification
\begin{equation}\label{V=UR}
V(\theta,\phi)=U\big(R_{\bi u(\phi)}(\theta) \big), \qquad
\bi u(\phi)=(-\sin\phi,\cos \phi,0)=\bi n(\pi/2,\phi+\pi/2);
\end{equation}
the unit vector $\bi n(\theta,\phi)$ is defined in \eqref{nTP}.
Moreover, the following decompositions hold:
\begin{equation}\label{Vdecomp}
U\big(R_{\bi n(\theta,\phi)}(\alpha)\big)=V(\theta,\phi)U(R_{\bi k}(\alpha))V(\theta,\phi)^\dagger, \qquad V(\theta,\phi)=\rme^{-\rmi\phi S_z}\rme^{-\rmi \theta S_y}\rme^{\rmi\phi S_z}.
\end{equation}

\subsection{Properties of the Wigner small-$d$-matrix}\label{W-d-m}
An explicit, but complicated, form of the Wigner small-$d$-matrix \eqref{Wmatrix} has been obtained \cite[(3.65)]{BieL81}; in particular, the explicit expressions for $s=1/2,\, 1 , \, 3/2, \, 2$ can be found in \cite[Fig.\ 44.1]{44.1}\footnote{The table can be downloaded from http://pdg.lbl.gov/2019/reviews/rpp2018-rev-clebsch-gordan-coefs.pdf}. From \cite[(3.65)]{BieL81} one sees that the form of the matrix elements is sufficiently simple when one of the indices takes the maximal value and one gets
\begin{equation}\label{dssm}
\abs{d^{(s)}_{s,m}(\theta)}^2=\frac{(2s)!}{(s+m)!\,(s-m)!}\left(\frac {1+x}2\right)^{s+m}\left(\frac {1-x}2\right)^{s- m}, \qquad x=\cos\theta;
\end{equation}
we reported only the square modulus, because we need only this, see \eqref{q+d^2}.

In general, the quantity $\abs{d^{(s)}_{\ell,m}(\theta)}^2$ is a polynomial in $\cos \theta$, as one sees from \cite[(3.72)]{BieL81}. Directly from the definition \eqref{Wmatrix} we have also
\begin{equation}\label{pi-m}
\abs{d^{(s)}_{\ell,m}(\theta)}^2=\abs{d^{(s)}_{-m,\ell}(\pi-\theta)}^2.
\end{equation}
The Wigner matrix turns out to be real and the following properties hold \cite[(3.80)-(3.82), (3.125)-(3.126)]{BieL81}:
\begin{equation}\label{dsymm}
d^{(s)}_{m',m}(\theta)=(-1)^{m-m'}d^{(s)}_{m,m'}(\theta)=d^{(s)}_{-m,-m'}(\theta),
\end{equation}
\begin{equation}\label{dsum}
\sum_{m=-s}^s d^{(s)}_{m_1,m}(\theta)d^{(s)}_{m_2,m}(\theta)=\sum_{m=-s}^s d^{(s)}_{m,m_1}(\theta)d^{(s)}_{m,m_2}(\theta)=\delta_{m_1,m_2}.
\end{equation}

\subsection{The $q$-coefficients}\label{app:qcoeff}
By using the expressions given in \cite[Fig.\ 44.1]{44.1} we can compute the $q$-coefficients in the cases $s=1/2,\,1,\,3/2$.

\subsubsection{Spin 1/2.}\label{app:1/2}
In this case we have $\abs{d_{\ell,h}^{(1/2)}(\theta)}^2=\frac 12 +2h\ell\cos\theta$. Then, from the definition \eqref{def:q} we obtain
\begin{equation}\label{q1/2}
q(m|\ell,h)=\frac 12 +2 \ell h m;
\end{equation}
we suppressed the index $\btheta$, because there is no arbitrariness in these indices, as recalled in Section \ref{joint1/2}. By \eqref{discr_s1}, \eqref{margq2}, we get \eqref{Mk1/2}.

\subsubsection{Spin 1.}\label{app:s=1}
In this case we have
\[
\abs{d_{0,0}^{(1)}(\theta)}^2=x^2, \qquad \abs{d_{0,\pm 1}^{(1)}(\theta)}^2=\abs{d_{\pm 1,0}^{(1)}(\theta)}^2=\frac{1-x^2}2,
\]
\[
\abs{d_{\pm 1,1}^{(1)}(\theta)}^2=\abs{d_{\mp 1,-1}^{(1)}(\theta)}^2=\frac{(1\pm x)^2}4, \qquad x:=\cos\theta.
\]
From \eqref{def:q}, by direct computations, we get the explicit expressions of the $q$-coefficients, with $a$ given in \eqref{s1_a}; by using this parameter as index, instead of $\btheta$, we have
\begin{equation}\label{q-1}\begin{split}
\quad &q_a(\pm1|1,\pm1)=q_a(\pm1|-1,\mp1)=1-\frac{(1+a)^3}8,
\\
&q_a(\mp1|1,\pm1)=q_a(\mp1|-1,\mp1)=\frac{(1-a)^3}8,
\\
&q_a(1|0,\pm1)=q_a(1|\pm1,0)=q_a(-1|0,\pm1)=q_a(-1|\pm1,0)=\frac{2+a}4\,(1-a)^2,
\\
&q_a(\pm 1|0,0)=\frac {1-a^3}2, \qquad q_a(0|0,\pm1)=q_a(0|\pm1,0)=\frac a2 \left(3-a^2\right),
\\
&q_a(0|0,0)=a^3,  \qquad \qquad q_a(0|1,\pm1)=q_a(0|-1,\mp1)=\frac a4\left(3+a^2\right). \end{split}
\end{equation}

\subsubsection{Spin 3/2.}\label{app:spin3/2}
In this case we have, with $x=\cos \theta$,
\[
\abs{d_{\pm 3/2,3/2}^{(3/2)}(\theta)}^2=\abs{d_{\mp 3/2,-3/2}^{(3/2)}(\theta)}^2=\frac {(1\pm x)^2}8,
\]
\[ 
\abs{d_{\pm 3/2,1/2}^{(3/2)}(\theta)}^2=\abs{d_{\mp 3/2,-1/2}^{(3/2)}(\theta)}^2=\abs{d_{1/2,\pm 3/2}^{(3/2)}(\theta)}^2= \abs{d_{-1/2,\mp 3/2}^{(3/2)}(\theta)}^2=\frac {3}8\left(1\pm x\right)(1-x^2),
\]
\[
\abs{d_{\pm 1/2,1/2}^{(3/2)}(\theta)}^2=\abs{d_{\mp 1/2,-1/2}^{(3/2)}(\theta)}^2=\frac {1\pm x}8\left(3x\mp 1\right)^2.
\]
From \eqref{def:q}, by direct computations, we get the explicit expressions of the $q$-coefficients, with $a$ given in Section \ref{joint3/2}; by using this parameter as index, instead of $\btheta$, we have
\begin{equation}\label{q3/2a}\begin{split}
&q_a(\pm 3/2|\pm 3/2,3/2)=\frac 1{16}\left(15-4a- 6a^2-4a^3- a^4\right),
\\
&q_a(\pm 3/2|\pm 3/2,-3/2)=\frac 1{16}\left(1-4a+6a^2-4a^3+ a^4\right),
\\
&q_a(\pm 3/2|\pm 3/2,1/2)=\frac 1{16}\left(11-12a- 6a^2+4a^3+ 3a^4\right),
\\
&q_a(\pm 3/2|\pm 3/2,-1/2)=\frac 1{16}\left(5-12a+ 6a^2+4a^3-3 a^4\right),\end{split}
\end{equation}
\begin{equation}\label{q3/2b}\begin{split}
&q_a(\pm 3/2|\pm1/2,1/2)=\frac 1{16}\left(7-4a+ 10a^2-4a^3-9 a^4\right),
\\
&q_a(\pm3/2|\pm1/2,-1/2)=\frac 1{16}\left(9-4a- 10a^2-4a^3+9 a^4\right),\end{split}
\end{equation}
\begin{equation}\label{q3/2c}\begin{split}
&q_a(\pm1/2|\pm3/2,3/2)=\frac a{16}\left(4+6a+ 4a^2+a^3\right),
\\
&q_a(\pm1/2|\pm3/2,-3/2)=\frac a{16}\left(4-6a+4a^2-a^3\right),
\\
&q_a(\pm1/2|\pm3/2,1/2)=\frac a{16}\left(12+6a- 4a^2-3a^3\right),
\\
&q_a(\pm1/2|\pm3/2,-1/2)=\frac a{16}\left(12-6a-4a^2+3a^3\right),\end{split}
\end{equation}
\begin{equation}\label{q3/2d}\begin{split}
&q_a(\pm1/2|\pm1/2,1/2)=\frac a{16}\left(4-10a+ 4a^2+9a^3\right),
\\
&q_a(\pm1/2|\pm1/2,-1/2)=\frac a{16}\left(4+10a+4a^2-9a^3\right);\end{split}
\end{equation}
the other coefficients are obtained by the symmetry properties \eqref{qsymm}.

\section{Orthogonal components}\label{app:orthcomp}

\subsection{Three orthogonal components}\label{app:3orth}
The set of  the three orthogonal spin components $\Acal_3=\{\Xo, \Yo, \Zo\}$ is invariant under the action of the order $24$ octahedron group $O\subset SO(3)$ \cite[Appendix B.4]{BGT18}, generated by the
$90^\circ$ rotations around the three coordinate axes: $S_{O}=
\{R_{\bi i}(\pi/2),\,R_{\bi  j}(\pi/2),\,R_{\bi  k}(\pi/2)\}$. Let us denote the three generators of $O$ by $g_1=R_{\bi
i}(\pi/2)$, $g_2=R_{\bi j}(\pi/2)$, $g_3=R_{\bi  k}(\pi/2)$; then we have the covariance relations
\begin{equation}\begin{split}
U_{g_1}
\Xo(x) U_{g_1}^\dagger = \Xo(x), \qquad
&U_{g_1}\Yo(y) U_{g_1}^\dagger  = \Zo(y), \qquad \ \ \
U_{g_1} \Zo(z) U_{g_1}^\dagger= \Yo(-z),
\\
U_{g_2}\Xo(x)U_{g_2}^\dagger  = \Zo(-x), \qquad
&U_{g_2} \Yo(y) U_{g_2}^\dagger = \Yo(y), \qquad \ \ \ U_{g_2}\Zo(z)U_{g_2}^\dagger  = \Xo(z),
\\
U_{g_3} \Xo(x) U_{g_3}^\dagger= \Yo(x), \qquad  &U_{g_3}  \Yo(y) U_{g_3}^\dagger = \Xo(-y),
\qquad  U_{g_3} \Zo(z)U_{g_3}^\dagger = \Zo(z).
\end{split}\end{equation}
Then, $\Mo\in \Mscr(\Acal_3)$ is a POVM on $\Xscr^3$ with the same covariance properties:
\begin{equation}\label{covMo3}\begin{split}
U_{g_1} \Mo(x,y,z) U_{g_1}^\dagger&= \Mo(x,-z,y),
\qquad
U_{g_2} \Mo(x,y,z) U_{g_2}^\dagger= \Mo(z,y,-x),
\\
&U_{g_3} \Mo(x,y,z) U_{g_3}^\dagger= \Mo(-y,x,z).\end{split}
\end{equation}

\subsection{Two orthogonal components} \label{app:2orth}
Here the set of target observables is $\Acal_2=\{\Xo,\Yo\}$. Their symmetry group is the dihedral group $D_4\subset SO(3)$, the order $8$
group of the $90^\circ$ rotations around the $\bi  k$-axis, together with the
$180^\circ$ rotations around $\bi  i$, $\bi  j$, $
\bi  n_1:=\bi n(\pi/2,\,\pi/4)$, and $\bi n_2:= \bi n(\pi/2,\, 3\pi/4)$.
Note that $D_4\subset O$.  The two rotations $S_{D_4}=
\{R_{\bi i}(\pi),R_{\bi  n_1}(\pi)\}$ generate $D_4$, as  we have
\[
R_{\bi j}(\pi)=R_{\bi  n_1}(\pi)R_{\bi  i}(\pi)R_{\bi  n_1}(\pi), \qquad R_{\bi
n_2}(\pi)=R_{\bi  i}(\pi)R_{\bi  n_1}(\pi)R_{\bi  i}(\pi),
\]
\[
R_{\bi k}(\pi/2)=R_{\bi  n_2}(\pi)R_{\bi  j}(\pi).
\]

As discussed in \cite[Appendix B.2]{BGT18}, the covariance relations are: $\forall(x,y)\in\Xscr^2$,
\begin{equation}\begin{split}
U\big(R_{\bi  i}(\pi)\big) \Xo(x) U\big(R_{\bi  i}(\pi)\big)^\dagger
= \Xo(x),
\qquad  &U\big(R_{\bi  i}(\pi)\big) \Yo(y) U\big(R_{\bi  i}(\pi)\big)^\dagger
= \Yo(-y),
\\
U\big(R_{\bi  n_1}(\pi)\big) \Xo(x) U\big(R_{\bi  n_1}(\pi)\big)^\dagger
= \Yo(x),
\qquad  &U\big(R_{\bi  n_1}(\pi)\big) \Yo(y) U\big(R_{\bi n_1}(\pi)\big)^\dagger
= \Xo(y).\end{split}
\end{equation}
Then, $\Mo\in \Mscr(\Acal_2)$ is a POVM on $\Xscr^2$ with the same covariance properties:
\begin{equation}\label{covMo2}\begin{split}
&U\big(R_{\bi  i}(\pi)\big) \Mo(x,y) U\big(R_{\bi  i}(\pi)\big)^\dagger
= \Mo(x,-y),
\\
&U\big(R_{\bi  n_1}(\pi)\big) \Mo(x,y) U\big(R_{\bi  n_1}(\pi)\big)^\dagger
= \Mo(y,x).\end{split}
\end{equation}

\subsection{Joint measurements from optimal cloning}\label{app:clon}
A technique to construct good multi-observables approximating a set of incompatible observables is based on optimal cloning \cite{Wer98,KW99,HSTZ14}; we already applied it to the context of MURs in \cite{BGT18}. Let us consider a system with Hilbert space $\Hscr$, of dimension $\dim(\Hscr)=d$, and let $\Scal(\Hscr)$ denote its state space; then,
the optimal approximate $\mathtt{r}$-cloning channel is the map
\[
\Phi : \Scal(\Hscr)\to\Scal(\Hscr^{\otimes \mathtt{r}}), \qquad \Phi(\rho) = \frac{d!\mathtt{r}!}{(d+\mathtt{r}-1)!}\,
\Pi_\mathtt{r} (\rho\otimes\openone^{\otimes (\mathtt{r}-1)}) \Pi_\mathtt{r} ,
\]
where $\Pi_\mathtt{r}$ is the orthogonal projection of $\Hscr^{\otimes \mathtt{r}}$ onto its
symmetric subspace ${\rm Sym}(\Hscr^{\otimes \mathtt{r}})$ \cite{KW99}.
Let $\{\Ao_1,\ldots, \Ao_\mathtt{r}\}$ be a set of observables, possibly incompatible; then, by using the adjoint channel we get the reasonably approximate multi-observable  $\Mo_{\rm cl} =
\Phi^*(\Ao_1\otimes\cdots\otimes\Ao_\mathtt{r})$, whose marginals are given by \cite{Wer98}
\begin{equation}\label{eq:maropt}
\Mo_{{\rm cl}[h]}(x)= \lambda_{d,\mathtt{r}}\Ao_h(x)+ \left(1-\lambda_{d,\mathtt{r}}\right)\frac\openone {d},\qquad \lambda_{d,\mathtt{r}}=\frac{d+\mathtt{r}}{\mathtt{r}\left(d+1\right)}.
\end{equation}
The multi-observable  $\Mo_{\rm cl}$ turns out to have the same symmetry properties of the set of observables $\{\Ao_1,\ldots, \Ao_\mathtt{r}\}$. Indeed, let $U$ be a unitary operator on $\Hscr$; by using the commutation property $U^{\otimes \mathtt{r}}\Pi_\mathtt{r}=\Pi_\mathtt{r}U^{\otimes \mathtt{r}}$, it is possible to prove the transformation rule
\[
U\Mo_{\rm cl}(x_1,\ldots,x_\mathtt{r})U^\dagger =\Phi^*\big(U\Ao_1(x_1)U^\dagger\otimes\cdots\otimes U\Ao_\mathtt{r}(x_\mathtt{r})U^\dagger\big).
\]

We shall use this construction for 2 or 3 orthogonal spin components; so, we have $d=2s+1$ and $\mathtt{r}=2,3$. The property above implies immediately that
$\Phi^*(\Xo,\Yo,\Zo)$ satisfies the covariance properties \eqref{covMo3} and $\Phi^*(\Xo,\Yo)$ the covariance properties \eqref{covMo2}.

\end{document}